\pdfoutput=1
\DocumentMetadata{}
\documentclass[acmsmall, screen]{acmart}

\newcommand{\out}{\mathtt{OUT}}

\usepackage{tcolorbox}
\usepackage{graphicx}   

\usepackage{hyperref}   
\usepackage{thmtools, thm-restate}
\usepackage[appendix=append,bibliography=common]{apxproof}
\usepackage{mdwlist}
\usepackage{mathtools}
\usepackage{enumitem}
\usepackage{xspace}
\usepackage{verbatim}   
\usepackage{xcolor,listings}
\usepackage{textcomp}
\usepackage{color}
\usepackage{makecell}
\usepackage[mathscr]{euscript}
\usepackage{tikz}
\usepackage{relsize}
\usepackage[edges]{forest}
\usepackage{booktabs}
\usepackage{tikz-qtree}
\usepackage{subcaption}
\usepackage{multirow}
\usetikzlibrary{arrows.meta}
\usetikzlibrary{positioning,calc}
\usetikzlibrary{decorations.text}
\usetikzlibrary{decorations.pathmorphing,decorations.pathreplacing}
\usetikzlibrary{arrows,petri, topaths,fit}
\usepackage[linesnumbered,algoruled, lined, noend]{algorithm2e}

\definecolor{codegreen}{rgb}{0,0.6,0}
\definecolor{codegray}{rgb}{0.5,0.5,0.5}
\definecolor{codepurple}{HTML}{C42043}
\definecolor{backcolour}{HTML}{F2F2F2}
\definecolor{bookColor}{cmyk}{0,0,0,0.90}  
\color{bookColor}

\lstdefinestyle{mystyle}{
    backgroundcolor=\color{backcolour},   
    commentstyle=\color{codegreen},
    keywordstyle=\color{codepurple},
    numberstyle=\numberstyle,
    stringstyle=\color{codepurple},
    basicstyle=\footnotesize\ttfamily
}
\usepackage[normalem]{ulem}

\newcommand{\shaleen}[1]{{\color{magenta} Shaleen: [{#1}]}}

\newcommand{\cut}[1]{}   
\makeatletter
\newcommand{\mytag}[2]{%
	\text{#1}%
	\@bsphack
	\protected@write\@auxout{}%
	{\string\newlabel{#2}{{#1}{\thepage}}}%
	\@esphack
}
\makeatother

\newcommand{\introparagraph}[1]{\noindent {\bf \em #1.}}  

\usepackage{aliascnt} 

\newtheorem{observation}{Observation}
\newtheoremrep{proposition}{Proposition}[section]
\newtheoremrep{lemma}{Lemma}[section]
\newtheoremrep{theorem}{Theorem}[section]

\providecommand{\bx}[0]{\mathbf{x}}

\providecommand{\bv}[0]{\mathbf{v}}

\providecommand{\mD}[0]{\mathcal{D}}
\providecommand{\mH}[0]{\mathcal{H}}

\providecommand{\mF}[0]{\mathcal{F}}
\providecommand{\mJ}[0]{\mathcal{J}}

\newcommand{\faq}{\mathsf{FAQ}}
\newcommand{\cq}{\mathsf{CQ}}
\newcommand{\bS}{\sigma}
\newcommand{\zerobf}{\mathbf{0}}
\newcommand{\onebf}{\mathbf{1}}

\providecommand{\mS}[0]{\mathcal{S}}
\providecommand{\mL}[0]{\mathcal{L}}
\providecommand{\mI}[0]{\mathcal{I}}

\providecommand{\tOUT}[0]{\texttt{OUT}}

\providecommand{\fhw}[0]{\texttt{fhw}}

\providecommand{\edges}[0]{\mathcal{E}}
\providecommand{\nodes}[0]{\mathcal{V}}

\providecommand{\domain}[0]{\mathbf{dom}}

\providecommand{\tree}[0]{\mathcal{T}}

\providecommand{\fhw}[0]{\mathsf{fhw}}
\providecommand{\subw}[0]{\mathsf{subw}}

\newcommand{\defeq}{\stackrel{\mathrm{def}}{=}}
\newcommand{\cd}{\ \leftarrow \ }

\providecommand{\mw}[0]{\mathsf{w}}
\providecommand{\mF}[0]{\mathcal{F}}
\providecommand{\mH}[0]{\mathcal{H}}
\providecommand{\mV}[0]{\mathcal{V}}
\providecommand{\mE}[0]{\mathcal{E}}

\providecommand{\eat}[1]{}

\newcommand{\sfreew}{\mathsf{freew}}

\providecommand{\wout}[0]{\texttt{pw}}

\AtBeginDocument{%
}

\newcommand{\revone}[1]{ #1}
\definecolor{reviewer2-color}{rgb}{0.8,0.0, 0.9}
\newcommand{\revtwo}[1]{ #1}
\newcommand{\revthree}[1]{ #1}

\setcopyright{acmlicensed}
\acmJournal{PACMMOD}
\acmYear{2024} \acmVolume{2} \acmNumber{5 (PODS)} \acmArticle{220} \acmMonth{11}\acmDOI{10.1145/3695838}

\received{May 2024}
\received[revised]{August 2024}
\received[accepted]{September 2024}

\begin{document}

\title{Output-sensitive Conjunctive Query Evaluation}
\author{Shaleen Deep}
\affiliation{%
  \institution{Microsoft Jim Gray Systems Lab}
  \country{USA}}
\email{shaleen.deep@microsoft.com}

\author{Hangdong Zhao}
\affiliation{%
  \department{Department of Computer Sciences}
  \institution{University of Wisconsin-Madison}
  \city{Madison}
  \country{USA}}
\email{hangdong@cs.wisc.edu}

\author{Austen Z. Fan}
\affiliation{%
  \department{Department of Computer Sciences}
  \institution{University of Wisconsin-Madison}
  \city{Madison}
  \country{USA}}
\email{afan@cs.wisc.edu}

\author{Paraschos Koutris}
\affiliation{%
  \department{Department of Computer Sciences}
  \institution{University of Wisconsin-Madison}
  \city{Madison}
  \country{USA}}
\email{paris@cs.wisc.edu}
\begin{abstract}
  Join evaluation is one of the most fundamental operations performed by database systems and arguably the most well-studied problem in the Database community. A staggering number of join algorithms have been developed, and commercial database engines use finely tuned join heuristics that take into account many factors including the selectivity of predicates, memory, IO, etc. However, most of the results have catered to either full join queries or non-full join queries but with degree constraints (such as PK-FK relationships) that makes join evaluation easier. Further, most of the algorithms are also not output-sensitive. 
  In this paper, we present a novel, output-sensitive algorithm for the evaluation of acyclic Conjunctive Queries (CQs) that contain arbitrary free variables. Our result is based on a novel generalization of the  Yannakakis algorithm and shows that it is possible to improve the running time guarantee of Yannakakis algorithm by a polynomial factor. Importantly, our algorithmic improvement does not depend on the use of fast matrix multiplication, as a recently proposed algorithm does. The application of our algorithm recovers known prior results and improves on known state-of-the-art results for common queries such as paths and stars. The upper bound is complemented with a matching lower bound \revthree{for star queries, a restricted subclass of acyclic CQs, and a family of cyclic CQs} conditioned on two variants of the $k$-clique conjecture. 
\end{abstract}
\maketitle
\section{Introduction}
Join query evaluation is the workhorse of commercial database systems supporting complex data analytics, data science, and machine learning tasks, to name a few. Decades of research in the database community (both theoretical and practical) have sought to improve the join query evaluation performance to speed up the processing over large datasets.

\looseness=-1
In terms of the theoretical guarantees offered by the state-of-the-art algorithms, several existing results are known for acyclic and cyclic queries. For acyclic queries, Yannakakis proposed an elegant framework for query evaluation~\cite{yannakakis1981algorithms}. Yannakakis showed that for a database $\mD$ of size $|\mD|$ and any acyclic Conjunctive Query (CQ) $Q(\bx_\mF)$ with output $\tOUT = Q(\mD)$ and free variables $\bx_\mF$, we can evaluate the query result in time $O(|\mD|+|\mD| \cdot |\tOUT|)$ in terms of data complexity. Depending on the structure of free variables in the query, it is possible to obtain a better runtime bound for the algorithm. For example, for the class of {\em free-connex acyclic} CQs, the algorithm is optimal with running time $O(|\mD| + |\tOUT|)$. Beyond acyclic queries, worst-case optimal joins (WCOJs) provide worst-case guarantees on query evaluation time for any full CQ (i.e., all variables are free). For non-full CQs, WCOJs are combined with the notion of tree decompositions to obtain tighter guarantees on the running time. In particular,~\cite{olteanu2015size} showed that any CQ can be evaluated in $O(|\mD|^\mw + |\tOUT|)$ time, where $\mw$ generalizes the fractional hypertree width from Boolean to CQs with arbitrary free variables. Abo Khamis et al.~\cite{abo2017shannon} showed that using the \textsf{PANDA} algorithm and combining it with the Yannakakis framework, any CQ can be evaluated in $\tilde{O}(|\mD|+|\mD|^\subw\ \cdot |\out|)$ time\footnote{$\tilde{O}$ notation hides a polylogarithmic factor in terms of $|\mD|$.}. Notably, the \textsf{PANDA} algorithm itself uses the Yannakakis algorithm as the final step for join evaluation once the cyclic query has been transformed into a set of acyclic queries. \revone{Berkholz and Schweikardt~\cite{berkholz2020constant} proposed an elegant width measure known as the \emph{free-connex submodular} width that allows additional enumeration related guarantees and thus, can evaluate certain classes of CQs \looseness=-1 efficiently.}

Despite its fundamental importance, no known improvements have been made to Yannakakis result until very recently. In an elegant result, Hu~\cite{hu2024fast} showed that using fast matrix multiplication (i.e., a matrix multiplication algorithm that multiplies two $n \times n$ matrices in $O(n^\omega)$ time, $\omega < 3$), we can evaluate any acyclic CQ in $O(|\mD|+ |\tOUT| + |\mD| \cdot |\tOUT|^{5/6})$ (if $\omega=2$, which gives the strongest result in their framework). The key idea is to use the star query as a fundamental primitive that benefits from fast matrix multiplication, an insight also explored by prior work~\cite{amossen2009faster, deep2020}, and evaluate any acyclic CQ in a bottom-up fashion by repeatedly applying the star query primitive. However, Hu's result has two principal limitations. First, \revtwo{the algorithm is non-combinatorial in nature, an undesirable property from a practical standpoint}. Second, the algorithm does not support general aggregations (aka FAQs cf.~\autoref{sec:agg}), a common usecase in SQL~\cite{date1989guide}. It is also open whether the upper bound obtained is optimal. Therefore, the question of whether Hu's upper bound can be improved and whether the improvement can be made using a combinatorial\footnote{Although combinatorial algorithm does not have a formal definition, it is intuitively used to mean that the algorithm does not use any algebraic structure properties. A key property of combinatorial algorithms is that they are practically efficient~\cite{williams2018subcubic}. Nearly all practical/production join algorithms known to date are combinatorial in nature.} algorithm \looseness=-1 remains open.

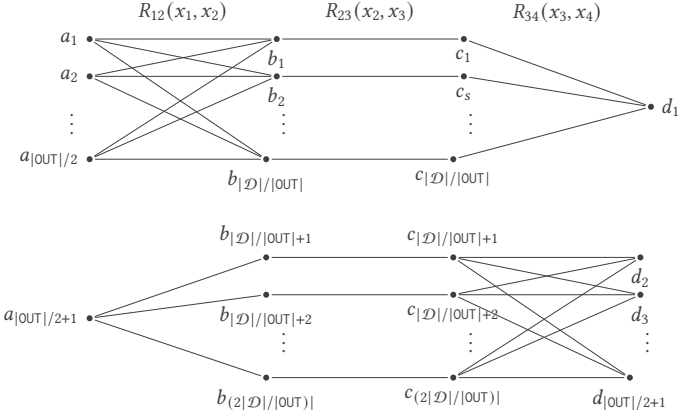
\begin{figure*}
    \begin{subfigure}{0.75\textwidth}
    \scalebox{0.8}{
        \def\x{1}
        \begin{tikzpicture}[
            mydot/.style={
              circle,
              fill,
              inner sep=1pt
            },
            >=latex,
            shorten >= 1pt,
            shorten <= 1pt
            ]
            \node[mydot,label={left:$a_1$},xshift=-2 cm] (a1) {}; 
            \node[mydot,below=of a1,yshift=0.5cm,label={left:$a_2$}] (a2) {}; 
            \node[below=of a2,yshift=0.5cm,label={left:$\vdots$}] (a3) {}; 
            \node[mydot,below=of a3,yshift=0.5cm,label={left:$a_{|\tOUT|/2}$}] (a4) {}; 
            
            \node[mydot,right=3 cm of a1,label={below:$b_1$}] (b1) {}; 
            \node[mydot,below=of b1,yshift=0.5cm,label={below:$b_2$}] (b2) {}; 
            \node[below=of b2,yshift=0.5cm,xshift=-10pt,label={right:$\vdots$}] (b3) {}; 
            \node[mydot,below=of b2,yshift=-0.25cm,label={below:$b_{|\mD|/{|\tOUT|}}$}] (b4) {};

            \node[mydot,right=2cm of b1,label={below:$c_1$},xshift=\x cm] (c1) {}; 
            \node[mydot,below=of c1,yshift=0.5cm,label={below:$c_2$}] (c2) {}; 
            \node[below=of c2,yshift=0.5cm] (cx) {}; 
            \node[below=of c2,yshift=0.5cm,xshift=-10pt,label={right:$\vdots$}] (c3) {}; 
            \node[mydot,below=of c2,yshift=-0.25cm,label={below:$\revone{c_{|\mD|/{2}}}$}] (c4) {};

            \node[mydot,right=2cm of c2,yshift=-0.5cm,label={right:$d_1$},xshift=\x cm] (d1) {}; 
            
            \path[] (a1) edge node[above, yshift=5pt]{$R^\mD_{12}(x_1, x_2)$} (b1)
              edge (b2)
              edge (b4);
            \path[] (a2) edge (b1)
              edge (b2)
              edge (b4);
            \path[] (a4) edge (b1)
              edge (b2)
              edge (b4);

            \path[] (b1) edge node[above, yshift=5pt]{$R^\mD_{23}(x_2, x_3)$} (c1);
            \path[] (b2) edge (c2);
            \path[] (b1) edge (c2);
            \path[] (b2) edge (cx);
            \path[] (b4) edge (c4); 

            \path[] (c1) edge node[above, yshift=20pt]{$R^\mD_{34}(x_3, x_4)$} (d1);
            \path[] (c2) edge (d1);
            \path[] (c4) edge (d1);

            \node[mydot,below=of a4,yshift=-1.5cm,label={left:$a_{|\tOUT|/2 + 1}$}] (a5) {};

            \node[mydot,below=of b4,yshift=-0.5cm,label={above:$b'_{1}$}] (b5) {}; 
            \node[mydot,below=of b5,yshift=0.5cm,label={below:$b'_{2}$}] (b6) {}; 
            \node[below=of b6,yshift=0.5cm,xshift=7pt,label={left:$\vdots$}] (b7) {}; 
            \node[below=of b6, yshift=0.5cm,] (bdots) {}; 
            \node[mydot,below=of b6,yshift=-0.25cm,label={below:$b'_{|\mD|/2}$}] (b8) {};

            \node[mydot,below=of c4,yshift=-0.5cm,label={above:$\revone{c'_1}$}] (c5) {}; 
            \node[mydot,below=of c5,yshift=0.5cm,label={below:$\revone{c'_2}$}] (c6) {}; 
            \node[below=of c6,yshift=0.5cm,xshift=-10pt,label={right:$\vdots$}] (c7) {}; 
            \node[mydot,below=of c6,yshift=-0.25cm,label={below:$\revone{c'_{|\mD|/|\tOUT|}}$}] (c8) {};

            \node[mydot,right=2cm of c5,label={below:$d'_1$},xshift=\x cm] (d5) {}; 
            \node[mydot,below=of d5,yshift=0.5cm,label={below:$d'_2$}] (d6) {}; 
            \node[below=of d6,yshift=0.5cm, xshift=-10pt,label={right:$\vdots$}] (d7) {}; 
            \node[mydot,below=of d6,yshift=-0.25cm,label={below:$d'_{|\tOUT|/2}$}] (d8) {};

            \path[] (a5) edge (b5);
            \path[] (a5) edge (b6);
            \path[] (a5) edge (b8); 

            \path[] (b5) edge (c5);
            \path[] (b6) edge (c5);
            \path[] (b6) edge (c6);
            \path[] (bdots) edge (c6);
            \path[] (b8) edge (c8); 

            \path[] (c5) edge (d5)
              edge (d6)
              edge (d8);

              \path[] (c6) edge (d5)
              edge (d6)
              edge (d8);

              \path[] (c8) edge (d5)
              edge (d6)
              edge (d8);
        \end{tikzpicture}}
    \end{subfigure}
    \caption{Database instance $\mD$ showing relational instances for relations $R_{12}(x_1, x_2), R_{23}(x_2, x_3),$ and $R_{34}(x_3, x_4)$ for the three path query.}
    \label{fig:threepath}
\end{figure*}

In this paper, we show that it is possible to evaluate any acyclic CQ $Q$ combinatorially in time $O(|\mD|+ |\tOUT|+|\mD| \cdot |\out|^{1-\epsilon})$ where $0 < \epsilon \leq 1$ is a query-dependent constant. An important class of queries that we consider is the path query $P_k(x_1, x_{k+1}) \leftarrow R_{12}(x_1, x_2) \wedge R_{23}(x_2, x_3)\wedge \dots \wedge R_{k, k+1}(x_k, x_{k+1})$. To give the reader an overview of our key ideas, we show our techniques on the  three path query $P_3(x_1, x_4)$.

\begin{example} \label{ex:paththree}
     Consider the database instance as shown in~\autoref{fig:threepath} for the $P_3(x_1, x_4)$ query. This example instance was also used by Hu to show the limitations of Yannakakis algorithm.
     \revone{For the relation instance $R^\mD_{23}$, each $b_i$ has a degree $|\tOUT|/2$ and is connected to $c_{1+(i-1) \cdot |\tOUT|/2 }, c_{2+(i-1) \cdot |\tOUT|/2}, \dots,$ $c_{i \cdot  |\tOUT|/2}$. Similarly, each $c'_i$ is connected to $b'_{1+(i-1) \cdot |\tOUT|/2}, b'_{2 + (i-1) \cdot |\tOUT|/2},$ $\dots, b'_{i \cdot  |\tOUT|/2}$ and has a degree $|\tOUT|/2$}. For any $1 \leq |\tOUT| \leq |\mD|$, Hu showed that any join order chosen by the Yannakakis algorithm for evaluating the query \revone{on the specific instance} must incur a materialization cost of $\Omega(|\mD| \cdot |\tOUT|)$. However, their argument assumes that the entire relation is used to do the join. We demonstrate an algorithm for this query that bypasses the assumption. First, partition $R^\mD_{12}(x_1, x_2)$ into $R_{12}^{\mD,H}(x_1, x_2)$ and $R_{12}^{\mD,L}(x_1, x_2)$ based on the degree threshold of $x_2$ values. In particular, \revone{for the instance under consideration, suppose we fix $\Delta = 2$}. We define:
$$ R_{12}^{\mD,L}(x_1, x_2) = \{ t \in R^\mD_{12} \mid |\sigma_{x_2 = t(x_2)} R^\mD_{12}| \leq \Delta\} $$
    
    Let $R_{12}^{\mD,H}(x_1, x_2) =  R^\mD \setminus R_{12}^{\mD, L}(x_1, x_2)$. Now, the original query can be answered by unioning the output of $P^H_3 = \pi_{x_1, x_4} (R_{12}^{\mD,H}(x_1, x_2) \wedge R^\mD_{23}(x_2, x_3) \wedge R^\mD_{34}(x_3, x_4))$ and $P^L_3 = \pi_{x_1, x_4} (R_{12}^{\mD,L}(x_1, x_2) \wedge R^\mD_{23}(x_2, x_3) \wedge R^\mD_{34}(x_3, x_4))$. Both of the join queries can be evaluated in $O(|\mD|)$ time, regardless of the value of $|\tOUT|$. {For both queries, we first apply a semijoin filter to remove all tuples from the input that do not participate in the join, a linear time operation}. $P^L_3$ is evaluated by first joining $R_{12}^{\mD,L}(x_1, x_2) \wedge R^\mD_{23}(x_2, x_3)$ \revone{using any standard join algorithm, and projecting the output on variables $x_1, x_3$. The materialized intermediate result is then joined with $R^\mD_{34}(x_3, x_4)$ and projected on $x_1, x_4$}. $P^H_3$ is evaluated in the opposite order by first joining $R^\mD_{23}(x_2, x_3) \wedge R^\mD_{34}(x_3, x_4)$ \revone{using any standard join algorithm, and projecting the output on variables $x_2, x_4$. The materialized intermediate result is then joined with } \revtwo{$R^{\mD,H}_{12}(x_1, x_2)$}\revone{, and finally projected on $x_1, x_4$}.
\end{example}

Building upon the idea in~\autoref{ex:paththree}, we show that by carefully partitioning the input and evaluating queries in a specific order, it is possible to improve the running complexity of join evaluation for a large class of CQs. 

\smallskip
\introparagraph{Our Contribution} In this paper, we develop a combinatorial algorithm for evaluating any CQ $Q$. We present a recursive algorithm that generalizes the Yannakakis algorithm and obtains a provable improvement in the running time over the Yannakakis algorithm. The generalization is clean enough for teaching at a graduate level. Our results are also easily extensible to support aggregations and commutative semirings. To characterize the running time of CQs, we also present a new simple width measure that we call the {\em projection width}, $\wout(Q)$, that closely relates to existing width measures known for acyclic queries. {As an example, our generalized algorithm when applied to the path query $P_k(x_1, x_{k+1})$ can evaluate it in time ${O}(|\mD|+ |\tOUT| + |\mD| \cdot |\tOUT|^{1-1/k})$. Rather surprisingly, this result (which is combinatorial and thus, assumes $\omega=3$) is already polynomially better than Hu's result for $k \leq 5$ even if we assume that $\omega=2$ to obtain the strongest possible result in their setting. At the heart of our main algorithm is a technical result that allows tighter bounding of the cost of materializing intermediate results when executing Yannakakis algorithm. We demonstrate the usefulness of our tighter bounding by proving that for path queries, we can further improve the running time to ${O}(|\mD|+ |\tOUT| + |\mD| \cdot |\tOUT|^{1-1/\lceil (k+1)/2 \rceil})$. For $k=3$, this result matches the lower bound of $\Omega(|\mD|+ |\tOUT| +|\mD| \cdot \sqrt{|\tOUT|})$ that holds for evaluating any $P_k$ for $k \geq 2$.
}
Further, we establish tightness of our results by demonstrating a matching lower bound on the running time for a \revthree{
restricted subclass} of acyclic and cyclic CQs. Our lower bounds are applicable to both join processing (aka Boolean semiring) and FAQs.

\eat{Worst case optimal joins (WCOJ) provide worst case guarantees on query evaluation time for any CQ. However, WCOJ are not output sensitive, i.e., the big-O complexity expression does not contain the output size as a parameter. Further, WCOJ do not provide any meaningful guarantees for queries that are not \emph{full}. For such queries, WCOJ are combined with the notion of tree decompositions to obtain tighter guarantees on the running time. In particular,~\cite{olteanu2015size} showed that using the idea of factorization, it is possible to evaluate any query in time $O(|D|^\fhw + |\out|)$ time. This result was later improved by~\cite{abo2017shannon} who showed that the running time can be further improved to $O(|D|^\subw\ + |\out|)$. These results have one principal limitation. First, the additive term only holds for full queries. If the query is not full, the $|\out|$ terms becomes multiplicative instead of additive. For non-full queries, both algorithms require an application of the celebrated Yannakakis algorithm. In~\cite{yannakakis1981algorithms}, Yannakakis showed that for any acyclic CQ $Q$ (not necessarily full) and database $D$, there exists an algorithm that can evaluate $Q(D)$ in time $O(|D| \cdot |\out|)$, where $\out$ denotes $Q(D)$. Remarkably, Yannakakis algorithm achieves optimal guarantees for the class of \emph{free-connex} queries where a tighter running time guarantee of $O(|D| + |\out|)$ can be achieved. However, for arbitrary acyclic CQs, it remains unknown if Yannakakis algorithm can be improved or not.

In this paper, we will show that it is indeed possible to evaluate every acyclic CQ $Q$ in time $O(|D| \cdot |\out|^{1-f(Q)})$ where $f(Q)$ is a constant smaller than one and depends only on the query $Q$. }
\section{Preliminaries and Notation}
\eat{
\hangdong{change into hypergraphs notations - i.e. the ones used in \textsf{PANDA} paper}

\hangdong{use $\chi(s) \cap \chi(t)$ for join variables}

\hangdong{use $P$ for head variables}

\hangdong{use $P_t$ for head variables present at the subtree rooted at $t$}}

\introparagraph{Conjunctive Queries} We associate a Conjunctive Query ($\cq$) $Q$ to a hypergraph
$\mH = (\mV,\mE)$, where $\mV = [n]=\{1,\ldots,n\}$ and $\mE$ is a set of hyperedges; each hyperedge is a subset of $[n]$. The relations (also refered to as atoms) of the query are $R_J(\bx_J), J \in \mE$, where $\bx_J = (x_j)_{j\in J}$ is the schema (or variables) of $R_J$, for any $J \subseteq [n]$.
The $\cq$ is:
\begin{equation}
   Q(\bx_{\mF}) \cd \bigwedge_{J \in \mE} R_J(\bx_J) \label{eqn:conjunctive:query}
\end{equation}
The variables in the head of the query $\bx_\mF$ ($\mF \subseteq [n]$) are the {\em free} variables (or the projection variables). A CQ is 
{\em full} if $\mF = [n]$, and it is {\em Boolean} 
if $\mF = \emptyset$, simply written as $Q()$. A database $\mD$ is a finite set of relations. \revtwo{We refer to the relational instance\footnote{We will frequently refer to the relational instance $R_J^\mD$ as just relation for the sake of brevity.} for relation $R_J$ in database $\mD$ using $R_J^\mD$.} The size $|R^\mD_J|$ of a relation $R^\mD_J$ is the number of its tuples.  The size of a database $|\mD| = \sum_{R_J \in \mD} |R^\mD_J|$. \revthree{A schema $\mathcal{X} = (x_1, \dots, x_g)$ of a relation is a non-empty tuple of distinct variables where each variable can takes values from domain $\domain$. We treat schemas and sets of variables interchangeably, assuming a fixed ordering of variables.} We will say that a variable $x \in \mV(\mH)$ that is present in at least two relations is a {\em join variable}, while a variable that is present in exactly one relation \looseness=-1 is {\em isolated}. 



\smallskip
\introparagraph{Tree Decomposition} A {\em tree decomposition} of a hypergraph $\mH = ([n], \mE)$ is a pair $(\tree, \chi)$ where $\tree$ is a tree and $\chi: V(\tree) \to 2^{[n]}$ maps each node
$t$ of the tree to a subset $\chi(t)$ of vertices such that

(1) Every hyperedge $F\in \mE$ is a subset of some $\chi(t)$, $t\in V(\tree)$,

(2) For every vertex $v \in [n]$,
    the set $\{t \mid v \in \chi(t)\}$ is a non-empty
    (connected) sub-tree of $\tree$. 

The sets $e = \chi(t)$ are called the {\em bags} of the tree decomposition. We use $\chi^{-1}(e)$ to recover the node of the tree with bag $e$. A query is said to be {\em $\alpha$-acyclic}\footnote{Throughout the paper, we will use acyclic to mean $\alpha$-acyclic.} if there exists a tree decomposition such that every bag $B$ of the tree decomposition corresponds uniquely to an input relation $R_B$. Such a tree decomposition is referred to as the {\em join tree} of the query. It is known that an $\alpha$-acyclic query can be evaluated in time $O(|\mD| + |\tOUT|)$. For simplicity, we will sometimes refer to the node of a join tree through the relation assigned to the node (or the relational instance) since there is a one-to-one mapping. A {\em rooted tree decomposition} is a tree decomposition that is rooted at some node $r \in V(\tree)$ (denoted by $(\tree, \chi, r)$. Different choices of $r$ change the orientation of the tree. We will use $\mL(\tree) \subseteq V(\tree)$ to denote the set of leaf nodes of a rooted tree decomposition and $\mI(\tree) = V(\tree) \setminus \mL(\tree)$ as the set of internal nodes (i.e. all non-leaf nodes) of the rooted tree decomposition. For any node $t \in V(\tree)$ in a rooted tree decomposition, we use $\mF_t$ to denote the set of free variables that appear in the subtree rooted at $t$ (including free variables in $\chi(t)$). The subtree rooted at $t$ is denoted via $\tree_t$.

A CQ $Q$ with free variables $\mF$ is called {\em free-connex acyclic} it is $\alpha$-acyclic and the hypergraph $([n], \mE \cup \{\mF\})$ is also $\alpha$-acyclic.

\smallskip
\introparagraph{Tuples and Operators} A \revthree{{\it tuple $v$}} over a set of variables $\bx_J$ is a total function that maps each variable $x \in \bx$ to a value in $\domain$. Given a tuple $v$ defined over $\bx$, and a set of variables $\mS \subseteq \bx$, $t(\mS)$ is the restriction of $t$ onto $\mS$. For a relation $R_J$ over variables $\bx_J$, $\mS \subseteq \bx_J$, and a tuple $s = v(\mS)$, we define $\sigma_{\mS = s} (R_J) = \{t \mid t \in R^\mD_J \land t(\mS) = s\}$ as the set of tuples in $R^\mD_J$ that agree with $s$ over variables in $\mS$, and $\pi_{\mS}(R_J) = \{t(\mS) \mid t \in R^\mD_J\}$ as the set of restriction of the tuples in $R^\mD_J$ to the variables in $\mS$. {The output or result of evaluating a CQ $Q$ over $\mD$ (denoted $Q(\mD)$) can be defined as $\{ \pi_{\bx_\mF}(v(\bx_{[n]})) \mid v(\bx_J) \in R^\mD_J, \forall J \in \edges  \}$.}
We denote by $\tOUT$ the result of running $Q$ over database $\mD$ (i.e., $\tOUT = Q(D)$) and we use $|\tOUT|$ to denote the number of tuples in $Q(D)$.

For relation $R_J$ over variables $\bx_J$, a threshold $\Delta$, and a set $\mS \subset \bx_J$, we say that a tuple $v(\mS)$ is {\it heavy} if $|\sigma_{\mS = v(\mS)} (R^\mD_J)| > \Delta$, and {\it light otherwise}. We will use $d(v, \mS, R^\mD_J)$ to denote $|\sigma_{\mS = v(\mS)} (R^\mD_J)|$. \revtwo{The \emph{semijoin} of two relations $R^\mD_{J_1}$ and $R^\mD_{J_2}$ (denoted as $R^\mD_{J_1} \ltimes R^\mD_{J_2}$) is defined as $\pi_{J_1} (R^\mD_{J_1} \land R^\mD_{J_2})$. A \emph{full reducer}~\cite{bernstein1981power} of a database $\mD$ is a finite sequence (that only depends on the schema of the relations in $\mD$) of semijoin operations that filters out tuples from the relations in the database $R^\mD_J$ such that $R^\mD_J = \pi_{J} (\bigwedge_{J \in \mE} R^\mD_J(\bx_J))$, i.e., all remaining tuples in the input relations participate in the result of the full join query $\bigwedge_{J \in \mE} R^\mD_J(\bx_J)$.} 

\smallskip
\introparagraph{Model of Computation} We use the standard RAM model with uniform cost measure. For an instance of size $N$, every register has length $O(\log N)$. Any arithmetic operation (such as addition,
subtraction, multiplication and division) on the values of two registers can be done in $O(1)$ time. Sorting the values of $N$ registers can be done in $O(N \log N)$ time.
\section{Projection Width}

Consider a CQ $Q$ with hypergraph $\mH$ and free variables  $\bx_\mF$.
We first define the {\em reduced query}  of $Q$, denoted $\textsf{red}[Q]$ (in~\cite{hu2024fast}, this is called a cleansed query) via Algorithm~\ref{alg:measure}. The while-loop of the procedure essentially follows the GYO algorithm~\cite{marc1979universal, yu1979algorithm}, with the difference that we can remove isolated variables only if they are not free. The reduced query is well-defined because the resulting hypergraph after the while-loop terminates is the same independent of the sequence of vertex and hyperedge removals.
If $\mF = V(\mH)$, the procedure does not remove any variables (but may potentially remove hyperedges). If $\mF = \emptyset$, the while-loop is identical to the GYO algorithm and will return the hypergraph $(\emptyset,\{\{\}\})$. We say that $Q$ is {\em reduced} if $Q = \textsf{red}[Q]$, i.e., the query cannot be further reduced.

\begin{algorithm}[!ht]
    \DontPrintSemicolon 
    \SetKwInOut{Input}{Input}
    \SetKwInOut{Output}{Output}
    \Input{acyclic $Q$ with hypergraph $\mH$ and free variables $\bx_\mF$}
    \Output{$\textsf{red}[Q]$}
    \SetKwFunction{subquery}{\textsc{processSubquery}}
    \SetKwProg{myproc}{\textsc{procedure}}{}{}
    \SetKwData{ret}{\textbf{return}}
    {$\mH' \leftarrow $multihypergraph of $\mH$ \tcc*{Edges $\mE'$ in $\mH'$ are a multiset }} 
    \While{$\mH'$ has changed}{
       \If{ $\exists$ isolated variable $x \notin \bx_\mF$}{
       \revthree{\ForEach{$e \in E(\mH')$}{
        $e \leftarrow e \setminus \{x\}$ \tcc*{remove $x$ from all hyperedges\label{line:remove:inedge}}
       }}
       $V(\mH') \gets V(\mH) \setminus \{x\}$ \tcc*{remove $x$ from the vertex set}\label{line:remove:vertex}}
       \If{ $\exists$ $e, f \in E(\mH')$ such that $e \subseteq f$}{$E(\mH') \gets E(\mH') \setminus \{e\}$ \label{line:remove:edge}}       
    }
    \ret $\mH',\mF$
    \caption{Reduced CQ} \label{alg:measure}
 \end{algorithm}

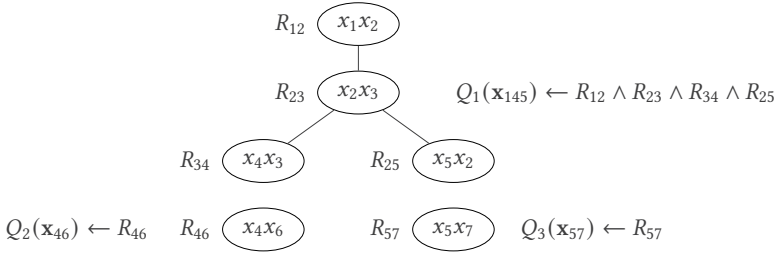
\begin{figure}
\eat{
    \begin{subfigure}{0.45\textwidth}
    \scalebox{0.77}{
	\centering
		\begin{tikzpicture}
                \node (A) [draw=black, ellipse, align=center] {$x_1 {x}_2$};
                \node [draw=none, left of=A] {$R_{12}$};
                \node (B) [draw=black, ellipse, below of=A] {${x}_2 {x}_3$};
                \node [draw=none, left of=B] {$R_{23}$};
                \draw[-] (A) edge node[draw=none] {} (B);
                \node (C) [draw=black, ellipse,below of=B, xshift = -40pt] {${x}_4 {x}_3$};
                \node (D) [draw=black, ellipse,below of=B, xshift = 40pt] {${x}_5 {x}_2$};
                \node [draw=none, left of=C] {$R_{34}$};
                \node [draw=none, left of=D] {$R_{25}$};
                \draw[-] (B) edge node[draw=none] {} (C);
                \draw[-] (B) edge node[draw=none] {} (D);
                \node (E) [draw=black, ellipse,below of=C] {${x}_4 {x}_6$};
                \node [draw=none, left of=E] {$R_{46}$};
                \node (F) [draw=black, ellipse,below of=D] {${x}_5 {x}_7$};
                \node [draw=none, left of=F] {$R_{57}$};
	\end{tikzpicture}}
	\caption{The graph $G^\exists_Q$.} \label{ex:complex}
    \end{subfigure}

    \begin{subfigure}{0.30\textwidth}
    \scalebox{0.77}{
	\centering
		\begin{tikzpicture}
                \node (A) [draw=black, ellipse, align=center] {${x}_2$};
                \node [draw=none, left of=A] {$R'_{12}$};
                \node (B) [draw=black, ellipse, below of=A] {${x}_2 {x}_3$};
                \node [draw=none, left of=B] {$R'_{23}$};
                \draw[-] (A) edge node[draw=none] {} (B);
                \node (C) [draw=black, ellipse,below of=B, xshift = -40pt] {${x}_3$};
                \node (D) [draw=black, ellipse,below of=B, xshift = 40pt] {${x}_2$};
                \node [draw=none, left of=C] {$R'_{34}$};
                \node [draw=none, left of=D] {$R'_{25}$};
                \draw[-] (B) edge node[draw=none] {} (C);
                \draw[-] (B) edge node[draw=none] {} (D);
                \node (E) [draw=black, ellipse,below of=C] {$ $};
                \node [draw=none, left of=E] {$R'_{46}$};
                \node (F) [draw=black, ellipse,below of=D] {$ $};
                \node [draw=none, left of=F] {$R'_{57}$};
	\end{tikzpicture}}
	\caption{Hypergraph after removing free variables.} \label{ex:complexfreeremoved}
    \end{subfigure}}
    \scalebox{0.99}{
	\centering
		\begin{tikzpicture}
                \node (A) [draw=black, ellipse, align=center] {$x_1 {x}_2$};
                \node [draw=none, left of=A] {$R_{12}$};
                \node (B) [draw=black, ellipse, below of=A] {${x}_2 {x}_3$};
                \node [draw=none, left of=B] {$R_{23}$};
                \node [draw=none, right of=B, xshift = 80pt] {$Q_1(\bx_{145}) \leftarrow R_{12} \wedge R_{23} \wedge R_{34} \wedge R_{25}$};
                \draw[-] (A) edge node[draw=none] {} (B);
                \node (C) [draw=black, ellipse,below of=B, xshift = -40pt] {${x}_4 {x}_3$};
                \node (D) [draw=black, ellipse,below of=B, xshift = 40pt] {${x}_5 {x}_2$};
                \node [draw=none, left of=C] {$R_{34}$};
                \node [draw=none, left of=D] {$R_{25}$};
                \draw[-] (B) edge node[draw=none] {} (C);
                \draw[-] (B) edge node[draw=none] {} (D);
                \node (E) [draw=black, ellipse,below of=C] {${x}_4 {x}_6$};
                \node [draw=none, left of=E] {$R_{46}$};
                \node (F) [draw=black, ellipse,below of=D] {${x}_5 {x}_7$};
                \node [draw=none, left of=F] {$R_{57}$};
                \node [draw=none, left of=E, xshift = -50pt] {$Q_2(\bx_{46}) \leftarrow R_{46}$};
                \node [draw=none, right of=F, xshift = 30pt] {$Q_3(\bx_{57}) \leftarrow R_{57}$};
	\end{tikzpicture}} 
    \caption{Depiction of the graph $G^\exists_Q$ and the decomposition of the running example query $Q(\bx_{14567}) \leftarrow R_{12}(\bx_{12}) \wedge R_{23}(\bx_{23}) \wedge R_{34}(\bx_{34}) \wedge R_{25}(\bx_{25}) \wedge R_{46}(\bx_{46}) \wedge R_{57}(\bx_{57})$} \label{ex:cc}
\end{figure}

The following \revone{three} properties of reduced queries will be important in this section.

\begin{propositionrep} \label{prop:reduced}
An acyclic $Q$ is free-connex acyclic if and only if all the variables of $\textsf{red}[Q]$ are free.
\end{propositionrep}
\begin{proof}
    Consider an acyclic CQ $Q$ such that all its variables in $\textsf{red}[Q]$ are free. Then, the hypergraph formed by adding a hyperedge containing all free variables is still acyclic, and thus satisfies the definition of free-connex acyclic. Indeed, there exists a join tree for the new hypergraph where the new hyperedge with all free variables is the root and all remaining hyperedges are its children. 

    For the other direction, consider an acyclic CQ $Q$ such that there exists a non-free variable (say $z$) in $\textsf{red}[Q]$. Clearly, such a variable is not isolated. Let $e_1$ and $e_2$  be two hyperedges that contain $z$. Let $E^*$ be the hyperedge added to the hypergraph containing all the variables. We claim that the new hypergraph containing edge $E^*$ is no longer acyclic. Indeed, suppose that $E^*$ is the root of the join tree (if there exists a join tree, we can reorient it to make any node as the root). Note that $e_1$ and $e_2$ cannot be subsets of each other and thus, one cannot be in the subtree of the other in the join tree. This is because if (say) $e_1$ contains a free variable (say $f$) that is not contained in $e_2$, then having $e_1$ in the subtree of $e_2$ will violate the variable connectedness condition for $f$, as $f$ is present in $E^*$ and $e_1$ but not in $e_2$. Thus, $e_1$ and $e_2$ must be either be directly connected to $E^*$ since both $e_1$ and $e_2$ may contain isolated free variables (all of which are also present in $E^*$), or they may be connected to the root through a series of intermediate nodes. However, the variable $z$ is not present in $E^*$ as $E^*$ only contains free variables and thus, the connectedness condition cannot be satisfied for variable $z$. Thus, a join tree cannot exist and the query cannot be free-connex acyclic.
\end{proof}
    
\begin{propositionrep} \label{prop:leaffree}
Let $Q$ be a reduced acyclic CQ, and $\tree$ be a join tree of $Q$. 
Then, every leaf node of $\tree$ has an isolated free variable.
\end{propositionrep}
  \begin{proof}
     Suppose there exists a leaf node $t$ that violates the desired property. First, we establish that $\chi(t)$ contains free variable(s). For the sake of contradiction, let us assume that $\chi(t)$ contains no free variable. Since all isolated variables have already been removed, it must be the case that all remaining variables in $\chi(t)$ are also present in its parent and thus are join variables. However, such hyperedges are removed by \autoref{line:remove:edge} of~\autoref{alg:measure}. Therefore, $\chi(t)$ must contain free variable(s).

     Next, suppose that the free variable is not isolated. Then, it must be the case that the free variable is also present in the parent (otherwise it would be isolated). By the same argument as before, since all other variables in $\chi(t)$ are also present in the parent, we get a contradiction since such hyperedges are removed by~\autoref{alg:measure}. This concludes the proof.
 \end{proof}

\revone{\begin{proposition} \label{prop:redprop}
Let $Q$ be a reduced acyclic CQ, and $\tree$ be a join tree of $Q$. Then, no variable in any node of $\tree$ can be isolated and non-free.
\end{proposition}

\autoref{prop:redprop} follows directly from the operation performed on line~\ref{line:remove:inedge}-\ref{line:remove:vertex} in~\autoref{alg:measure}.

}

Second, we define the {\em decomposition} of a CQ $Q$ following~\cite{hu2024fast}. Define the graph $G^\exists_Q$, where each hyperedge is a vertex, and there is an edge between $e,e'$ if they share a non-free variable. Let $E_1, \dots, E_k$ be the connected components of $G^\exists_Q$. Then, the decomposition of $Q$, denoted $\textsf{decomp}(Q)$, is a set of queries $\{Q_1, \dots, Q_k\}$, where $Q_i$ is the CQ with hypergraph $(\bigcup_{e \in E_i} e, E_i)$ and free variables $\mF \cap \bigcup_{e \in E_i} e$. If the decomposition of $Q$ has exactly one query, we say that $Q$ is {\em existentially connected}.

\begin{definition}[Projection Width] \label{def:reduced}
Consider an acyclic CQ $Q$. Then, $\wout(Q)$ is the maximum number of relations across all queries in $\textsf{decomp}(\textsf{red}[Q])$.
\end{definition}




\begin{example} \label{ex:star}
Consider the query 
$$Q^\star_\ell(x_1, \dots, x_\ell) \leftarrow R_1(x_1,y) \wedge \dots \wedge R_k(x_\ell,y).$$
One can observe that $\textsf{red}[Q^\star_\ell] = Q^\star_\ell$. The decomposition of the reduced query has only one connected component with $\ell$ atoms, hence $\wout(Q^\star_\ell) = \ell$.
\end{example}

\begin{example}
Consider the query $Q(\bx_{14567}) \leftarrow R_{12}(\bx_{12}) \wedge R_{23}(\bx_{23}) \wedge R_{34}(\bx_{34}) \wedge R_{25}(\bx_{25}) \wedge R_{46}(\bx_{46}) \wedge R_{57}(\bx_{57})$. Note the query is already reduced since there are no isolated variables or any hyperedges that are contained in another. To compute the projection width of the query, \autoref{ex:cc} shows the graph $G^\exists_Q$ where each atom of the query becomes a vertex. Note that there is no edge between $R_{34}$ and $R_{46}$ because $x_4$ is a free variable. This graph has three connected components, and the largest connected component contains four hyperedges, hence $\wout(Q(\bx_{14567})) = 4$. \autoref{ex:cc} also shows the three queries corresponding to each connected component of $G^\exists_Q$.
\end{example}

 We observe that $1 \leq \wout(Q) \leq |E(\mH)|$, and also that $\wout(Q)$ is always an integer. When $Q$ is full, every hyperedge of the reduced hypergraph forms its own connected component of size one; in this case, $\wout(Q)=1$. More generally:

 \begin{propositionrep}
 An acyclic CQ $Q$ is free-connex acyclic if and only if $\wout(Q)=1$.
 \end{propositionrep}
 \begin{proof}
Suppose that $Q$ is free-connex. From Proposition~\ref{prop:reduced}, $\textsf{red}[Q]$ has only free variables. Hence, in the decomposition every hyperedge forms a connected component of size one. Thus, $\wout(Q)=1$.

For the other direction, suppose that $\wout(Q)=1$. Let us examine the reduced query $Q' = \textsf{red}[Q]$. Then, every query in the decomposition of $Q'$ has size one. However, this implies in turn that all variables of $Q'$ are free; indeed, if there exists a non-free variable $x$, the variable $x$ would be isolated and that would violate the fact that $Q'$ is reduced.  From Proposition~\ref{prop:reduced}, we now have that $Q$ is free-connex acyclic.
\end{proof}

Hu~\cite{hu2024fast} defined a notion similar to $\wout(Q)$, called free-width, $\textsf{freew}(Q)$. To compute free-width, we first define the free-width of an existentially connected query to be the size of the smallest set of hyperedges that covers all the isolated variables. Then, $\textsf{freew}(Q)$ is the maximum freewidth over all queries in the decomposition of $Q$. It is easy to see that $\textsf{freew}(Q) \leq \wout(Q)$.


  \section{Main Result}

In this section, we describe our main result that builds upon the celebrated Yannakakis algorithm. First, we recall the algorithm and its properties as outlined in~\cite{yannakakis1981algorithms}. \revtwo{The first step of the algorithm is to apply the full reducer which removes all tuples from the input database relations that do not contribute to the output.} \revthree{The main idea of the algorithm is to use a bottom-up evaluation strategy over the join tree. In particular, a node $s$ is processed once each of its children have been processed. Consider node $s$ whose children are all leaves. The key step of the algorithm is to do the join of relational instances corresponding to $s$ and each of its children but projecting the output on only the free variables in the subtree and the variables on node $s$. The bottom-up process continues until we reach root and final join result is returned.}

Yannakakis showed two key properties of~\autoref{alg:yann}. First, when the processing of a node $s$ is over (i.e. when the algorithm has finished execution of \autoref{line:j} for node $s$), it holds that 
$$T^\mD_{\chi(s)} = \pi_{\mF_{s} \cup \chi(s)} (\wedge_{B \in V(\tree_{s})} R^\mD_{\chi(B)})$$

Therefore, when only node $r$ is left in the tree, then $T^\mD_{\chi(r)} = \pi_{\mF_{r} \cup \chi(r)} (\wedge_{B \in V(\tree)} R^\mD_{\chi(B)})$, and thus, $\pi_{\bx_\mF} (T^\mD_{\chi(r)})$ gives the desired result. Second, Yannakakis showed that the entire algorithm takes time $O(|\mD|+|\mD| \cdot |\tOUT|)$, which is the time taken to execute~\autoref{line:j} in each iteration.

\begin{algorithm}[!ht]
    \DontPrintSemicolon 
    \SetKwInOut{Input}{Input}
    \SetKwInOut{Output}{Output}
    \Input{acyclic query $Q(\bx_\mF)$, database instance $\mD$, rooted join tree $(\tree, \chi,r)$}
    \Output{$Q(\mD)$}
    \SetKwFunction{subquery}{\textsc{processSubquery}}
    \SetKwProg{myproc}{\textsc{procedure}}{}{}
    \SetKwData{ret}{\textbf{return}}
    $\mD := (R^\mD_{\chi(s)})_{s \in V(\tree)} \leftarrow $ \textbf{apply} a full reducer for $\mD$ \label{line:reduce}\;
    $K \leftarrow \text{a queue of } V(\tree)$ following a post-order traversal of $\tree$ \\
    \While{$K \neq \emptyset$ }{
        $s \leftarrow K.\textit{pop}()$ \;
        \If{$s$ is a leaf \revone{in $\tree$}}{
            $T^\mD_{\chi(s)} = \pi_{{\chi(s)}} (R^\mD_{\chi(s)}(\bx_{\chi(s)}))$ \label{line:trueleaf}\\
        } \Else{
            $T^\mD_{\chi(s) \cup \mF_s} = \pi_{{\chi(s) \cup \mF_s}} \bigg(R^\mD_{\chi(s)} \wedge \left(\bigwedge_{(s, t) \in E(\tree)} \pi_{\mF_t \cup (\chi(s) \cap \chi(t))} T^\mD_{\chi(t)} \right) \bigg)$ \label{line:j} \\
            $\chi(s) \leftarrow \chi(s) \cup \mF_s$ \label{line:end} \\
        }
    }
    \ret $\pi_{\mF}(T^\mD_{\chi(r)})$
    \caption{Yannakakis Algorithm~\cite{yannakakis1981algorithms}} \label{alg:yann}
 \end{algorithm}

\subsection{Output-sensitive Yannakakis}

This section presents a general lemma that forms the basis of our main result. In particular, we will show that under certain restrictions of the instance and the root of the join tree, \autoref{alg:yann} (Yannakakis algorithm) achieves a better runtime by a factor of $\Delta$. In the next section, we will present an algorithm that takes advantage of this observation. 

The first condition is that the root node must contain an isolated free variable. This is not necessarily true for any node in the join tree, but we can make it happen by choosing the root to be a leaf (recall that in a reduced instance, every leaf node has an isolated free variable). The second condition is that all tuples in the root node over the join variables are heavy w.r.t. a degree threshold $\Delta$. This is not generally true (unless $\Delta=1$), so we will need to partition the instance to achieve this requirement. For a join tree $(\tree,\chi)$, we define $\chi^\Join(s) \subseteq \chi(s)$ to return only the join variables of $\chi(s)$ \eat{and $\chi^{IF}(s) \subseteq \chi(s)$ to return the isolated free variables of $\chi(s)$. For a reduced query, it holds that $\chi^\Join(s) \cap \chi^{IF}(s) = \emptyset$ and \autoref{prop:redprop} tells us that $\chi^\Join(s) \cup \chi^{IF}(s) = \chi(s)$.}

\begin{lemma} \label{lem:basic}
Let $Q(\bx_\mF)$ be a reduced acyclic CQ, $(\tree, \chi,r)$ a rooted join tree of $Q$, and $\mD$ be a database instance. Suppose that
    \begin{enumerate}
        \item $R_{\chi(r)}$ has at least one isolated free variable
        \item for every $v \in R^\mD_{\chi(r)}$, we have $d(v, \chi^\Join(r), R^\mD_{\chi(r)}) > \Delta$ for some integer $\Delta \geq 1 $.
    \end{enumerate}
Then, \autoref{alg:yann} runs in time $O(|\mD| + (\sum_{t \in \mI(\tree)} |R^\mD_{\chi(t)}|) \cdot |\tOUT|/\Delta)$.
\end{lemma}
\begin{proof}

\begin{figure}[t]
	\centering
  \scalebox{0.9}{
		\begin{tikzpicture}
                \node (V) [draw=black, ellipse, minimum width=50pt,align=center,label={left:$R_{\chi(r)}$}] {$\chi(r)$};
                \node (E) [draw=none, below of=V] {$\vdots$};
                \draw[->] (V) edge node[draw=none] {} (E);
                \node (F) [draw=black, ellipse, below of=E, minimum width=50pt,align=center,label={left:$R_{\chi(s)}$}] {$\chi(s)$};
                \node (G) [draw=black, ellipse, below of=F, xshift=-120pt, yshift=-20pt, minimum width=50pt,align=center,label={left:$T_{\chi(t_1)}$}] {$\chi(t_1)$};
                \draw[->] (F) edge node[draw=none, xshift=-5pt, label=left:{\small$\chi(t_1) \cap \chi(s)$}] {} (G) ;
                \node (H) [draw=black, ellipse, below of=F, yshift=-20pt,minimum width=50pt,align=center,label={left:$T_{\chi(t_2)}$}] {$\chi(t_2)$};
                \draw[->] (F) edge node[draw=none,xshift=5pt, label=left:{\small$\chi(t_2) \cap \chi(s)$}] {} (H);
                \node (I) [draw=none, right of=H, xshift=25pt] {$\dots$};
                \node (J) [draw=black, ellipse, right of=I, xshift=40pt, minimum width=50pt,align=center,label={left:$T_{\chi(t_k)}$}] {$\chi(t_k)$};
                \draw[->] (F) edge node[draw=none, label=right:{\small$\chi(t_k) \cap \chi(s)$}, minimum width=50pt,xshift=10pt] {} (J);
	\end{tikzpicture}}
	\caption{A join tree with root node $r$. Edge labels show the common variables between bag $\chi(s)$ and bag $\chi(t_i)$.} \label{fig:jt}
\end{figure}

To prove this result, we will show that the size of the intermediate result that is materialized in~\autoref{line:j} for an internal node $s$ is bounded by $|\tOUT| \cdot (|R^\mD_{\chi(s)}| / \Delta)$.

Consider the point in the algorithm where we join the node $s$ with its children nodes $t_1, \dots, t_k$, as shown in Figure~\ref{fig:jt}. \revthree{The relational instances assigned to the nodes $t_1, \dots, t_k$ may not necessarily correspond to the base relations $R^\mD_{\chi(t_i)}$ since a previous iteration of the while loop may have performed a join that led to creation of the intermediate relation $T^\mD_{\chi(t_i)}$.} Let $v$ be a tuple over the variables $\bx_{[n]}$ such that $v(\bx_\mF) \in Q(\mD)$ \revthree{and for each relation $R^\mD_J(\bx_J)$, it holds that $v(\bx_J) \in R^\mD_J$}. We first claim the following inequality:
    
\begin{equation} \label{eq:one}
    \prod_{i \in [k]}d(v, \chi(t_i) \cap \chi(s), T^\mD_{\chi(t_i)}) \leq |\tOUT|/\Delta
\end{equation} 
    
To show this, let $Z$ be \revthree{the set of all variables in the nodes $r,t_1, \dots, t_k$ except for the isolated free variables. Consider the join query $Q'(\bx_\mF)$ where the values for all variables except $Z$ have been fixed according to $v$ (i.e. the tuples in the relations are filtered as shown below). 
\begin{equation} \label{eq:qprime}
Q'(\mD) = \pi_{\bx_{\mF}}((R^\mD_{\chi(r)} \ltimes v(\chi^\Join(r))) \wedge_{i \in [k]} (T^\mD_{\chi(t_i)} \ltimes v(\chi^\Join(t_i))) \wedge_{u \in V(\tree) \setminus \{r, t_1,  \dots, t_k \}} (R^\mD_{\chi(u)} \ltimes v))
\end{equation} 

\looseness=-1
From the definition of $v$, it holds that $Q'(\mD) \subseteq Q(\mD)$. Next, we claim that $|Q'(\mD)|$ \looseness=-1 is exactly:

%
$$A_v := d(v, \chi^\Join(r), R^\mD_{\chi(r)}) \cdot \prod_{i \in [k]}d(v, \chi(t_i) \cap \chi(s), T^\mD_{\chi(t_i)}) $$

To see why this holds, first, observe that the semijoin of $v$ with all relations except for the root node and nodes $t_1, \dots, t_k$ (i.e. all the relations considered by the last conjunct of \autoref{eq:qprime}) fixes their size to one. This is because of our choice of $v$ that guarantees that $v(\bx_\mF) \in Q(\mD)$, and thus implies that the tuple formed by restricting $v$ onto the schema of each relation is present in the corresponding relational instance. Intuitively, it means that there exists a \emph{join path} from the root node to the leaf nodes $t_1, \dots, t_k$ since for each intermediate node, there is a tuple formed by restricting $v$ present in the corresponding input relation.

Next, note that $R_{\chi(r)}$ has at least one isolated free variable, and tuple $v(Z)$ fixes values for all join variables in the relation. Similarly, tuple $v(Z)$ also fixes values for all join variables in $R_{\chi(t_i)}$, which are $\chi(t_i) \cap \chi(s)$. Since the query is reduced, \autoref{prop:leaffree} guarantees that every leaf node has an isolated free variable. Further, \autoref{prop:redprop} guarantees that for all bags, every variable is either an isolated free variable or a join variable. Therefore, once the join variables of nodes $r, t_1, \dots, t_k$ have been fixed, it is guaranteed that all remaining variables in those nodes are output variables. Since we have already established that for all intermediate relations from root to the leaf nodes $t_i$, there exists tuples that join with $v(Z)$, it holds that $|Q'(\mD)| = |R^\mD_{\chi(r)} \ltimes v(\chi^\Join(r))| \cdot \Pi_{i \in [k]} |T^\mD_{\chi(t_i)} \ltimes v(\chi^\Join(t_i))| = A_v$.

To finish the claim, note that {$A_v = |Q'(\mD)| \leq |Q(\mD)| = |\tOUT|$} and also from assumption (2) of the lemma, we have that $d(v, \chi^\Join(r), R^\mD_{\chi(r)}) > \Delta$.}

\smallskip
To complete the proof, let $W = \chi(s) \cap (\bigcup_{i \in [k]} \chi(t_i))$  
denote the join variables of node $s$ that are also present in some leaf node $t_1, \dots, t_k$. We observe that the size of the intermediate bag  $\pi_{\chi(s) \cup \mF_s} (R^\mD_{\chi(s)} \wedge  T^\mD_{\chi(t_1)} \wedge \dots \wedge T^\mD_{\chi(t_k)})$ can be bounded by
    \begin{align*}
        &\sum_{w \in \pi_{W} (R^\mD_{\chi(s)})} |R^\mD_{\chi(s)} \ltimes w| \cdot \prod_{i \in [k]} |T^\mD_{\chi(t_i)} \ltimes w| \\
       & = \sum_{w \in \pi_{W} (R^\mD_{\chi(s)})} d(w, W, R^\mD_{\chi(s)}) \cdot \prod_{i \in [k]} d(w, \chi(t_i) \cap \chi(s), T^\mD_{\chi(t_i)}) \\ 
       &\leq \frac{|\tOUT|}{\Delta} \cdot \sum_{w \in \pi_{W} (R^\mD_{\chi(s)})} d(w, W, R^\mD_{\chi(s)}) \tag*{(using~\autoref{eq:one})} \\ 
       &\leq \frac{|\tOUT|}{\Delta} \cdot |R^\mD_{\chi(s)}| \tag*{(sum of all degrees is equal to the relation size)}
    \end{align*}

\revthree{Here, the first line bounds the total join size as the sum of sizes of the cartesian product of the relations after semijoin for each fixing of $w \in \pi_{W} (R^\mD_{\chi(s)})$. The first inequality holds since the degree product bound holds for all tuples $w$. Indeed, since the query is acyclic, once a full reducer has been applied, for each $w$, there exists a tuple $v$ over $\bx_{[n]}$ such that $w = v(W)$, $v(\bx_\mF) \in Q(D)$, and $v(\bx_J) \in R^\mD_J$ for each relation in $\mD$}. Observe that the time bound requires only using the size of the relation assigned to the parent of the leaf nodes (and not the sizes of the relations assigned to the leaf nodes itself).
\end{proof}

\looseness=-1
\introparagraph{Discussion} Observe that \autoref{lem:basic} degenerates to the standard bound of Yannakakis algorithm for $\Delta = 1$. However, as we will show in the next part, choosing $\Delta > 1$ can lead to better overall join processing algorithms. It is also interesting to note that the running time bound obtained in \autoref{lem:basic} can only be achieved when the root of the decomposition is a relation with the two properties as outlined in the statement. It can be shown that no other choice of the root node achieves a time better than $O(|\mD| \cdot |\tOUT|)$. \revone{It is important to note that \autoref{lem:basic} requires the query to be reduced. Indeed, without the reduced query requirement, relations may contain variables that are neither free and nor join. The presence of such variables renders the join size \looseness=-1 computation incorrect.}

 \subsection{Our Algorithm}
In this section, we will present the algorithm for our main result. We will first consider a CQ $Q$ that is existentially connected and reduced.

Algorithm~\ref{alg:our} shows the improved procedure. It processes the nodes of the join tree in a leaf-to-root order just like Yannakakis algorithm, \revone{i.e., a node $s$ is processed only after each of its children have been processed}. However, our algorithm departs in the operations involved in the processing of each \revone{node}. In particular, consider node $s$ whose children are all leaves. For every leaf node $t$ that is a child of $s$, we partition the relation $T^\mD_\chi(t)$ assigned to the node into two disjoint partitions (the heavy partition $T^{\mD,H}_\chi(t)$ and the light partition $T^{\mD,L}_\chi(t)$) using a chosen degree threshold. Then, we use \autoref{lem:basic} to process $T^{\mD,H}_\chi(t)$ and add the produced output into a set $\mJ$. This processing is done in lines~\ref{call:lemone}-\ref{line:end} of the while loop. The crucial detail is that the processing of the heavy partition is done by reorienting the join tree to be rooted at $t$, and then applying Yannakakis in a bottom-up \looseness=-1 fashion. 

 Once all heavy sub-relations of the leaf nodes are processed, we are left with all light parts of the relations for the children of node $s$. At this point, we join the relations in the subtree rooted at $s$ (line~\ref{line:join}) and then remove the nodes for children of $s$ from the join tree $\tree$(line~\ref{line:del}). \revone{The step on line~\ref{line:join} is identical to the one on line~\ref{line:j} of the Yannakakis algorithm. Note that modifying the structure of the join tree $\tree$ by deleting nodes in our algorthm is a departure from Yannakakis algorithm. Although modification of $\tree$ is not required for correctness of our algorithm, as we will show later, by controlling the node processing order $K$ (on line~\ref{line:order}), one can use our algorithm in innovative ways. Therefore, when $K$ is a subset of $V(\tree)$ instead of containing all the nodes, it is important to keep the database $\mD$ and $\tree$ up-to-date to reflect which nodes have been processed.} We next state the main result \looseness=-1 (the proof can be found in the appendix).

\begin{algorithm}[!ht]
    \DontPrintSemicolon 
    \SetKwInOut{Input}{Input}
    \SetKwInOut{OInput}{Optional Input}
    \SetKwInOut{Output}{Output}
    \Input{\revthree{reduced and existentially connected} acyclic query $Q(\bx_\mF)$, instance $\mD$, rooted join tree  $(\tree, \chi)$, $1 \leq \Delta \leq |\mD|$}
    \Output{$Q(\mD), (\tree, \chi), \mD$}
    \SetKwFunction{subquery}{\textsc{processSubquery}}
    \SetKwFunction{heavyleaf}{\textsc{processHeavyLeaf}}
    \SetKwFunction{lightleaf}{\textsc{processLightLeaves}}
    \SetKwProg{myproc}{\textsc{procedure}}{}{}
    \SetKwData{ret}{\textbf{return}}
    \textbf{choose} an arbitrary root $r$ for $\tree$ \\
     $\mD := (R^\mD_{\chi(s)})_{s \in V(\tree)} \leftarrow $ \textbf{apply} a full reducer for $\mD$ \\ 
     \revone{$\tree^{IN} \leftarrow$ clone of $\tree$ \tcc*{clone of the join tree since we will edit $\tree$ in-place}}
     $N \leftarrow |\mD|$ \tcc*{storing the size of the input}
        $K \leftarrow \text{a \revone{queue} of } s \in V(\tree)$ following a post-order traversal of $\tree$  \label{line:order}\;
    \While{$K \neq \emptyset$ }{
     $s \leftarrow K.\textit{pop}()$ \\
      \If( \label{line:leaf}){$s$ is a leaf \revone{in $\tree$}}{
            \If{$s$ is not the root of $\tree$} {
            	\revone{\If{$s$ is a leaf in $\tree^{IN}$} {
             \tcc{identical to line~\ref{line:trueleaf} in Alg~\ref{alg:yann} and initializes the $T^\mD_{\chi(s)}$}
            		$T^\mD_{\chi(s)} = \pi_{\chi(s)} (R^\mD_{\chi(s)}(\bx_{\chi(s)}))$ 
            	}}
                $\Delta_s \leftarrow \Delta \cdot (|T^\mD_{\chi(s)}| + N)/ N$ \label{line:delta}\;
                $T_{\chi(s)}^{\mD,H} = \{\bv \in T^\mD_{\chi(s)} \mid |\sigma_{ \bv(\chi^\Join(s))}(T^\mD_{\chi(s)})| > \Delta_s \}, \; T_{\chi(s)}^{\mD,L} = T^\mD_{\chi(s)} \setminus T_{\chi(s)}^{\mD,H}$ \;  \label{disjoint}
                $\mD^H_s \leftarrow (\mD \setminus T^\mD_{\chi(s)}) \cup T_{\chi(s)}^{\mD,H}$ \;
                \textbf{let} $Q_s(\mD^H_s)$ be the output of \autoref{alg:yann} with input as $\mD^H_s$ and $\tree$ rooted at $s$ \label{call:lemone}\; 
                $\mathcal{J} \leftarrow \mathcal{J} \cup Q_s(\mD^H_s)$ \label{union1}\;
                $\mD \leftarrow $ \textbf{apply} a full reducer for $(\mD \setminus T^\mD_{\chi(s)}) \cup T_{\chi(s)}^{\mD,L}$ \label{line:end} \;
                $T^\mD_{\chi(s)} =  \pi_{\chi(s)} (T_{\chi(s)}^{\mD,L})$ \;
            } \Else {
            $\mathcal{J} \leftarrow \mathcal{J} \cup \pi_{\mF}(T^\mD_{\chi(r)})$ \label{union2} \tcc*{$s$ is the only node in the tree}
            }
      } \Else(\label{line:nonleaf}){
            $T^\mD_{\chi(s) \cup \mF_s} = \pi_{\chi(s) \cup \mF_s} \bigg(T^\mD_{\chi(s)} \wedge \left(\bigwedge_{(s, t) \in E(\tree)} \pi_{\mF_t \cup (\chi(s) \cap \chi(t))} T^\mD_{\chi(t)} \right) \bigg)$ \label{line:join} \;
            $\chi(s) \leftarrow \chi(s) \cup \mF_s$ \;
             \textbf{truncate} all children of $s$ and directed edges $(s, t) \in E(\tree)$ from $\tree$ \label{line:del} \;
             $K.\textit{\revone{push\_to\_head}}(s)$ \label{push}\tcc*{$s$ became a leaf,  process $s$ immediately}
      }
    }    
    \ret $\mathcal{J}, (\tree, \chi), \mD$ \tcc*{Return output, tree decomposition, and instance}
    \caption{Generalized Yannakakis Algorithm} \label{alg:our}
 \end{algorithm}

 \begin{lemmarep} \label{thm:main}
     Given an acyclic join query $Q(\bx_\mF)$ that is existentially connected and reduced, database $\mD$, and an integer threshold $1 \leq \Delta \leq |\mD|$, \autoref{alg:our} computes the join result $Q(\mD)$ in time $O(|\mD| \cdot \Delta^{k-1} + |\mD| \cdot |\tOUT|/\Delta)$, where $k$ is the number of atoms in $Q$.
 \end{lemmarep}
 \begin{proof}
     We begin with a simple observation: the only step in the algorithm that modifies the join tree is the join operation on line~\ref{line:join}. In any given iteration of the while loop, once the join on line~\ref{line:join} is computed, node $s$ becomes a leaf node in the join tree since we delete all its children right after. Let $N$ denote the size of the input database\footnote{We do not use $|\mD|$ in the proof to avoid any confusion since the algorithm modifies $\mD$ repeatedly.}. Then, we get the following.
     
     \begin{observation} \label{obv:one}
         After every iteration of the while loop, it holds that all relations corresponding to the internal nodes of the join tree have size $O(N)$.
     \end{observation}

     Let us now bound the time required to process all leaf nodes. Fix a leaf node $s$. Except the operation on line~\ref{call:lemone}, all other operations take at most $O(N)$ time (note that $\mD$ is also modified in every iteration of the inner loop). We claim that  line~\ref{call:lemone} takes time $O(N \cdot |\tOUT|/\Delta)$. By our choice of root node (which is $s$) for calling Lemma~\ref{lem:basic}, $\tree$ is rooted at $s$. Further, note that the leaves of $(\tree,s)$ are the same leaves as $(\tree,r)$ (except for $s$, which is now the root). Further, by~\autoref{prop:leaffree}, it is also guaranteed that the bag of $s$ contains isolated free variable(s). Thus, the conditions of applying Lemma~\ref{lem:basic} are satisfied. Since our choice of threshold for the partition is $\Delta_s$, Lemma~\ref{lem:basic} tells us that the evaluation time required is at most big-O of

     $$\left(\sum_{B \in \mI(\tree)} |T^\mD_{\chi(B)}|\right)  \cdot |\tOUT|/\Delta_s = (|T^\mD_\chi(s)| + N) \cdot |\tOUT|/\Delta_s = N \cdot |\tOUT|/\Delta $$

     Here, the first equality holds because of~\autoref{obv:one} which guarantees that only $s$ can have $\Omega(N)$ size. The reader can verify that $\Delta \leq \Delta_s$, and thus, is well-defined.

     Next, we will bound the time required for the join operation on line~\ref{line:join}. For any node $s$, let $\#s$ denote the number of nodes in the subtree rooted at $s$ (including $s$). We will show by induction that time required for the join is $O(N \cdot \Delta^{\#s-1})$. 

     \textbf{Base Case.} In the base case, consider the leaf nodes with relations of size $N$. Consider such a leaf $\ell$ and note that the relation size satisfies $O(N \cdot \Delta^{\#\ell - 1}) = O(N)$ since the number of nodes a subtree rooted at a leaf is one (the leaf itself). 

     \textbf{Inductive Case.} Now, consider a node $s$ that is not a leaf.  The size of the relation for its child node $t$ is $O(N \cdot \tau^{\#t-1})$. For $t$, the degree threshold for partitioning the relation is $\Delta_t = \Delta \cdot (N \cdot \Delta^{\#t-1} +N) / N \leq 2 \Delta^{\#t}$. Therefore, by making the same argument as in the proof of \autoref{lem:basic}, we get the degree product as $\Pi_{t \in \text{child of } s} \Delta_{t} = O(\Delta^{\sum_{t \in \text{child of } s} \#t}) = \Delta^{\#s-1}$. Thus, the join time and the size of $R_{\chi(s)}$ is  $O(N \cdot \Delta^{\#s-1})$. Since the tree contains $|\tree|$ number of nodes, the total time required is dominated by the last iteration, giving us $O(N \cdot \Delta^{|\tree|-1})$.

     In each iteration of the while loop, one node of the tree is processed. Thus, the total number of iterations is $|\tree|=k$ and the total running time of the algorithm is $O(N \cdot \Delta^{k-1} + N \cdot |\tOUT|/\Delta)$.
 \end{proof}

  \begin{figure}
     \begin{subfigure}{0.30\textwidth}
     \centering
    \scalebox{0.77}{
		\begin{tikzpicture}
                \node (A) [draw=black, ellipse, align=center] {$x_1 {x}_2$};
                \node [draw=none, left of=A] {$R_{12}$};
                \node (B) [draw=black, ellipse, below of=A] {${x}_2 {x}_3$};
                \node [draw=none, left of=B] {$R_{23}$};
                \node [draw=none, above of=A, yshift = -10pt] {$Q_1(\bx_{145}) \leftarrow R_{12} \wedge R_{23} \wedge R_{34} \wedge R_{25}$};
                \draw[-] (A) edge node[draw=none] {} (B);
                \node (C) [draw=black, ellipse,below of=B, xshift = -40pt] {${x}_3 {x}_4$};
                \node (D) [draw=black, ellipse,below of=B, xshift = 40pt] {${x}_5 {x}_2$};
                \node [draw=none, left of=C] {$R_{34}$};
                \node [draw=none, left of=D] {$R_{25}$};
                \draw[-] (B) edge node[draw=none] {} (C);
                \draw[-] (B) edge node[draw=none] {} (D);
	\end{tikzpicture}}
	\caption{The Join tree for query ${Q_1}$} \label{fig:inputcc}
    \end{subfigure}
    \hfill
    \begin{subfigure}{0.30\textwidth}
    \centering
    \scalebox{0.77}{
		\begin{tikzpicture}
                \node (A) [draw=black, ellipse, align=center] {$x_3 {x}_4$};
                \node [draw=none, left of=A] {$R^H_{34}$};
                \node (B) [draw=black, ellipse, below of=A] {${x}_2 {x}_3$};
                \node [draw=none, left of=B] {$R_{23}$};
                \draw[-] (A) edge node[draw=none] {} (B);
                \node (C) [draw=black, ellipse,below of=B, xshift = -40pt] {${x}_1 {x}_2$};
                \node (D) [draw=black, ellipse,below of=B, xshift = 40pt] {${x}_5 {x}_2$};
                \node [draw=none, left of=C] {$R_{12}$};
                \node [draw=none, left of=D] {$R_{25}$};
                \draw[-] (B) edge node[draw=none] {} (C);
                \draw[-] (B) edge node[draw=none] {} (D);
	\end{tikzpicture}}
	\caption{Applying~\autoref{lem:basic} with $R^H_{34}$ as root node.} \label{fig:lemoneone}
    \end{subfigure}
    \hfill
    \begin{subfigure}{0.30\textwidth}
    \centering
    \scalebox{0.77}{
		\begin{tikzpicture}
                \node (A) [draw=black, ellipse, align=center] {$x_2 {x}_5$};
                \node [draw=none, left of=A] {$R^H_{25}$};
                \node (B) [draw=black, ellipse, below of=A] {${x}_2 {x}_3$};
                \node [draw=none, left of=B] {$R_{23}$};
                \draw[-] (A) edge node[draw=none] {} (B);
                \node (C) [draw=black, ellipse,below of=B, xshift = -40pt] {${x}_1 {x}_2$};
                \node (D) [draw=black, ellipse,below of=B, xshift = 40pt] {${x}_3 {x}_4$};
                \node [draw=none, left of=C] {$R_{12}$};
                \node [draw=none, left of=D] {$R^L_{34}$};
                \draw[-] (B) edge node[draw=none] {} (C);
                \draw[-] (B) edge node[draw=none] {} (D);
	\end{tikzpicture}}
	\caption{Applying~\autoref{lem:basic} with $R^H_{25}$ as root node.} \label{fig:lemonetwo}
    \end{subfigure}

    \begin{subfigure}{0.30\textwidth}
	\centering
    \scalebox{0.77}{
		\begin{tikzpicture}
                \node (A) [draw=black, ellipse, align=center] {$x_1 {x}_2$};
                \node [draw=none, left of=A] {$R_{12}$};
                \node (B) [draw=black, ellipse, below of=A] {${x}_2 {x}_3$};
                \node [draw=none, left of=B] {$R_{23}$};
                \draw[-] (A) edge node[draw=none] {} (B);
                \node (C) [draw=black, ellipse,below of=B, xshift = -40pt] {${x}_4 {x}_3$};
                \node (D) [draw=black, ellipse,below of=B, xshift = 40pt] {${x}_5 {x}_2$};
                \node [draw=none, left of=C] {$R^L_{34}$};
                \node [draw=none, left of=D] {$R^L_{25}$};
                \draw[-] (B) edge node[draw=none] {} (C);
                \draw[-] (B) edge node[draw=none] {} (D);
                \node (rect) at (B) [draw,thick,dashed,minimum width=5.5cm,minimum height=2.15cm,yshift=-0.5cm,xshift=-0.2cm] {};
                \node [draw=none, left of=rect,xshift=-1.2cm, yshift=0.8cm] {$R_{245}$};
	\end{tikzpicture}}
	\caption{Joining the relations in subtree rooted at $R_{23}$} \label{fig:light}
    \end{subfigure}
    \hfill
    \begin{subfigure}{0.30\textwidth}
    \centering
    \scalebox{0.77}{
		\begin{tikzpicture}
                \node (A) [draw=black, ellipse, align=center] {$x_2 {x}_4 x_5$};
                \node [draw=none, left of=A,xshift=-10pt] {$R^H_{245}$};
                \node (B) [draw=black, ellipse, below of=A] {${x}_1 {x}_2$};
                \node [draw=none, left of=B] {$R_{12}$};
                \draw[-] (A) edge node[draw=none] {} (B);
	\end{tikzpicture}}
	\caption{Applying~\autoref{lem:basic} with $R^H_{245}$ as root node.} \label{fig:lemonethree}
    \end{subfigure}
    \hfill
    \begin{subfigure}{0.30\textwidth}
	\centering
    \scalebox{0.77}{
		\begin{tikzpicture}
                \node (A) [draw=black, ellipse, align=center] {$x_1 {x}_2$};
                \node [draw=none, left of=A] {$R_{12}$};
                \node (B) [draw=black, ellipse, below of=A] {${x}_2 {x}_4 x_5$};
                \node [draw=none, left of=B,xshift=-10pt] {$R^L_{245}$};
                \draw[-] (A) edge node[draw=none] {} (B);
	\end{tikzpicture}}
	\caption{Joining the remaining relations.}  \label{fig:final}
    \end{subfigure}
     \caption{Evaluating the running example query using~\autoref{alg:our}. Each figure shows a rooted join tree.} 
     \label{fig:enter-label}
 \end{figure}

\eat{\shaleen{ \tiny Hangdong: I think line 5 should have a queue instead of a stack and line 25 should push $s$ to head. We previously had a stack instead of a queue and just push on line 26. I don't think that works (a stack of post order traversal will have the root at the head). This needs to be checked.}}

 \begin{example}
     We use the query $Q_1$ shown in~\autoref{fig:inputcc} as an example to demonstrate the execution of~\autoref{alg:our}. {The join tree will be visited in the order $K = \{\chi^{-1}(\bx_{34}), \chi^{-1}(\bx_{25}), \chi^{-1}(\bx_{23}), \chi^{-1}(\bx_{12})\}$.}
     We first visit node for $R_{34}$ and since it is a leaf node, we apply~\autoref{lem:basic} with the root node as the heavy partition $R^{H}_{34}$ as shown in \autoref{fig:lemoneone}. Once the heavy partition has been processed, we replace $R^\mD_{34}$ with $R^{\mD,L}_{34}$ in the input database (\autoref{line:end}), which will be used in all subsequent iterations of the algorithm. Next, we process leaf node $R_{25}$ by again calling \autoref{lem:basic} with heavy partition $R^{H}_{25}$ as the root. $R^\mD_{25}$ is then replaced with $R^{\mD,L}_{25}$ in $\mD$.

     Now, both leaf nodes have their relations replaced by the light partitions. When we process node $R_{23}$, a non-leaf relation, we join all relations in the subtree rooted at $R_{23}$ (shown in dashed rectangle in \autoref{fig:light}). Thus, the query $Q'(\bx_{245}) \leftarrow R^L_{34} \wedge R^L_{25} \wedge R_{23}$ is evaluated, the variables in the bag for the node are replaced with $\bx_{245}$ and the relation is the output of the query $Q'(\bx_{245})$. Leaf nodes $R^L_{25}$ and $R^L_{34}$ are deleted, and $R_{245}$ becomes a leaf node. {$\chi^{-1}(\bx_{245})$ is added to the front of $K$ on line~\ref{push}.} Therefore, in the next iteration, we take the heavy partition of node  $R^{H}_{245}$ and apply \autoref{lem:basic}. Finally, we visit the node for $R_{12}$, process the join $Q''(\bx_{145}) \leftarrow R_{12} \wedge R^L_{245}$ (\autoref{fig:final}) and the root node bag is modified to $\bx_{145}$ with relation as the result $R^\mD_{145} = Q''(\mD)$. Node $R^L_{245}$ is deleted. At this point $K = \emptyset$, and we union $Q''(\bx_{145})$ and $\mJ$ on \autoref{union2}. Since the entire tree is now processed, the while loop terminates, and the final result $\mJ$ is returned.
 \end{example}

 \smallskip
 \introparagraph{Finding the optimal threshold $\Delta$} To find the optimal threshold that minimizes the running time of \autoref{alg:our}, we can equate the two terms in the running time expression of \autoref{thm:main} to obtain $\Delta = |\tOUT|^{1/k}$, giving us the running time as $O(|\mD| \cdot |\tOUT|^{1 - 1/k})$. However, the value of $|\tOUT|$ is not known apriori. To remedy this issue, we use the \emph{doubling trick}~\cite{auer1995gambling} that was first introduced in the context of multi-armed bandit algorithms. The key idea is to guess the value of $|\tOUT|$. {Suppose the guessed output size is $O$ and let $\alpha$ be a constant value that is an upper bound of the constant hidden in the big-O runtime complexity of~\autoref{alg:our}. We start with an estimate of $O_1 = 2^0$ in the first round and run the algorithm. If the algorithm does not finish execution in $\alpha \cdot |\mD| \cdot O_i^{1 - 1/k}$ steps, then we terminate the algorithm and pick the new estimate to be $O_{i+1} = 2 \cdot O_{i}$ and re-run the algorithm. However, if the algorithm finishes, then we have successfully computed the query result. Note that $\alpha$ can be determined by doing an analysis of the program and counting the number of RAM model operations required for each line. It is easy to see that the algorithm will terminate within \revtwo{$\lceil \log (2 \cdot |\tOUT|) \rceil$} rounds and the total running time is $\alpha \cdot |\mD| \cdot \sum_{i \in [\lceil \log_2 (2 \cdot |\tOUT|) \rceil]} O_i^{1 - 1/k} = \alpha \cdot |\mD| \cdot \sum_{i \in [ \lceil \log_2 (2 \cdot |\tOUT|) \rceil]} 2^{i \cdot (1 - 1/k)} = O(|\mD| \cdot |\tOUT|^{1 - 1/k})$ for any $k \geq 2$.} Formally:
 
  \begin{theorem} \label{thm:ecr}
    Given a reduced and existentially connected acyclic query $Q(\bx_\mF)$, and a database $\mD$, we can compute $Q(\mD)$ in time $O(|\mD| + |\mD| \cdot |\tOUT|^{1 - 1/k})$, where $k$ is the number of atoms in $Q$.
 \end{theorem}
 
 We note that the doubling trick argument for join evaluation has been used in prior works~\cite{deng2018overlap, amossen2009faster}. This idea can also be applied to the results in~\cite{hu2024fast}, allowing us to shave the polylog factors in the total running time and removing the need to estimate the output size via sophisticated algorithms.

 \smallskip
 \introparagraph{General CQs} Finally, we discuss what happens for a general acyclic CQ that may not be reduced or existentially connected. In this case, we can relate the runtime of the algorithm to the projection width we defined in the previous section. \revthree{The main insight here is that once a general acyclic CQ has been reduced and decomposed, we can evaluate each component separately. The query evaluation output of each component can be combined easily since the query can now be viewed as a free-connex acyclic query, whose evaluation is well understood~\cite{bagan2007acyclic}.}
 Note that for a CQ that is reduced and existentially connected, $\wout$ is exactly the number of atoms in the query.

\begin{theoremrep} \label{thm:all}
Given an acyclic CQ $Q$ and a database $\mD$, we can compute the output $Q(\mD)$ in time $O(|\mD| + |\tOUT|+|\mD| \cdot |\tOUT|^{1 - 1/\wout(Q)})$.
\end{theoremrep}
\begin{proof}
This result follows exactly the argument in~\cite{hu2024fast}, which shows that the runtime of an evaluation algorithm for $Q$ can be reduced to the computation of each query in $\textsf{decomp}(\textsf{red}[Q])$. In particular, as a first step it can be shown that we can compute in linear time an instance $\mD'$ such that $Q(\mD) \leftarrow \textsf{red}[Q](\mD')$; this is done by doing semijoins in the same order as the GYO algorithm.  

As a second step, we use~\autoref{thm:ecr} to compute each query $Q_i \in \textsf{decomp}(\textsf{red}[Q])$ in time $O(|\mD| + |\mD| \cdot |Q_i(\mD)|^{1 - 1/k_i})$, where $k_i$ is the number of atoms in $Q_i$. Then, we materialize each query and compute $Q(\mD)$ by doing the join $\bigwedge_i Q_i(\mD)$. This join query corresponds to a free-connex acyclic CQ, so it can be evaluated in time $O(|\mD| + \sum_i |Q_i(\mD)| + |Q(\mD)|) = O(|\mD| + |\tOUT|)$. The desired claim follows from the fact that projection width is defined as the maximum number of atoms in any query $Q_i$ of the decomposition.
\end{proof}

\eat{
  \begin{algorithm}[!ht]
    \DontPrintSemicolon 
    \SetKwInOut{Input}{Input}
    \SetKwInOut{Output}{Output}
    \Input{acyclic hypergraph $\mH$, free variables $\mF$, Database $\mD$}
    \Output{Query result}
    \SetKwFunction{subquery}{\textsc{processSubquery}}
    \SetKwProg{myproc}{\textsc{procedure}}{}{}
    \SetKwData{ret}{\textbf{return}}
    $\mD \leftarrow $ \textbf{apply} a full reducer for $\mD$ \;
    \ForEach{$Q \in \textsf{decomp}(\textsf{red}[Q])$}{
        $R^C(\mF_C) \leftarrow Q^C(D)$ \label{join:cc}
    }
    \ret $\pi_{\mF} (\wedge_{C \in \{C_1, \dots, C_\ell\}} R^C(\mF_C))$ \label{line:full} \tcc*{Query is a full acyclic CQ}
    \caption{CQ Evaluation} \label{alg:overall}
 \end{algorithm}
 }

\revtwo{
 \introparagraph{Self-Joins} So far, we have assumed that the query does not contain any repeated relations (i.e. no self-joins). However, our framework can handle self-joins as well by performing a few basic transformations. First, if a query contains repeated relations, we make copies of the input relation(s) in the database involved in the self-join, assign a unique relational name to each copy, and use the unique name for each occurrence of the repeated relation to rewrite the query. Then, we order the schema of each relation according to the variable order $[n]$. This operation is straightforward since reordering of the variables in the schema of a relation merely corresponds to shuffling each tuple in the relation to match the reordered schema. Finally, for all atoms $R_J(\bx_J), S_J(\bx_J), \dots, V_J(\bx_J)$ that have the same schema, we only keep one atom in query (say $R_J(\bx_J)$) and modify the instance $R^D_J = R^D_J \cap S^D_J \cap \dots \cap V^D_J$. Each step takes at most $O(|\mD|)$ time and satisfies the formulation of a CQ as defined in \looseness=-1 \autoref{eqn:conjunctive:query}, and thus our main result can extend to self-joins.

 \begin{example}
     Consider the query $Q(x_1, x_2, x_4) \leftarrow R(x_1, x_2) \wedge R(x_2, x_1) \wedge S(x_2, x_3) \wedge S(x_3, x_4)$. The query contains a self-join on both $R$ and $S$. Therefore, we first create two copies of relation $R$: $R_1(x_1, x_2)$ and $R_2(x_2, x_1)$; and two copies of $S$: $S_1(x_2, x_3)$ and $S_2(x_3, x_4)$. The rewritten query becomes $Q(x_1, x_2, x_4) \leftarrow R_1(x_1, x_2) \wedge R_2(x_2, x_1) \wedge S_1(x_2, x_3) \wedge S_2(x_3, x_4)$. Next, we modify the schema of relation $R_2$ to $R_2(x_1, x_2)$ to order the variables in relation $R^\mD_2(x_1, x_2)$. Finally, since $R_1$ and $R_2$ have the same schema, we compute $R^\mD_1 = R^\mD_1 \cap R^\mD_2$ and discard $R_b$. The final rewritten query is $Q(x_1, x_2, x_4) \leftarrow R_1(x_1, x_2) \wedge S_1(x_2, x_3) \wedge S_2(x_3, x_4)$.
 \end{example}
 
 Repeated variables in the schema of a relation can also be handled in a linear time preprocessing step by modifying the schema to only have one occurrence of each variable and modifying each tuple in the relational instance.
 }

\section{Extension to Aggregation} \label{sec:agg}

\introparagraph{Semirings}  A tuple $\bS = (\boldsymbol{D}, \oplus, \otimes, \mathbf{0}, \mathbf{1})$ is a (commutative) \textit{semiring} if $\oplus$ and $\otimes$ are binary operators over $\boldsymbol{D}$ for which:
\begin{enumerate}
    \item $(\boldsymbol{D}, \oplus, \zerobf)$ is a commutative monoid with additive 
    identity $\zerobf$ (i.e., $\oplus$ is associative and commutative, and $a \oplus 0 = a$ for all $a \in D$);
    \item $(\boldsymbol{D}, \otimes, \onebf)$ is a (commutative) monoid with multiplicative 
    identity $\onebf$ for $\otimes$;
    \item $\otimes$ distributes over $\oplus$, i.e., $a \otimes (b \oplus c) = (a \otimes b) \oplus (a \otimes c)$ for $a, b, c \in \boldsymbol{D}$; and 
    \item $a \otimes \zerobf = \zerobf$ for all $a \in \boldsymbol{D}$.
\end{enumerate}
Such examples include the Boolean semiring $\mathbb{B}$ = $(\{ \textsf{false}, \textsf{true} \} , \vee, \wedge, \textsf{false}, \textsf{true})$, the natural numbers semiring $(\mathbb{N}, +, \cdot, 0, 1)$, and the \textit{tropical semiring} $\mathsf{Trop}^{+}=\left(\mathbb{R}_{+} \cup\{\infty\}, \min ,+, \infty, 0\right)$.

\introparagraph{Functional Aggregate Queries}
Green et al.~\cite{Green07} developed the idea of using annotations of a semiring to reason about provenance over natural joins. That is, every relation $R_J$ is now a $\bS$-relation, i.e., each tuple in $R_J$ is annotated by an element from the domain $\boldsymbol{D}$ of $\bS$. Tuples not in the relation are annotated by $\zerobf \in \boldsymbol{D}$ implicitly. Standard relations are essentially $\mathbb{B}$-relations. Abo Khamis et al.~\cite{FAQ16} introduced the functional aggregate queries ($\faq$) that express join-aggregate queries via semiring annotations. A $\faq$ $\varphi$ (over a semiring $\bS$) is the following:
\begin{equation} \label{eq:sumprod-rule}
     \varphi(\bx_\mF) \leftarrow \bigoplus_{\bx_{[n] \setminus \mF}} \bigotimes_{J \in \mE} R_J(\bx_J),
\end{equation}
where $(i)$ $([n], \mE)$ is the \textit{associated hypergraph} of $\varphi$ (the hyperedges $\mathcal{E} \subseteq 2^{[n]}$), $(ii)$ $\bx_\mF$ are the head variables ($\mF \subseteq [n]$), and $(iii)$ each $R_J$ is an input $\bS$-relation of schema $\bx_J$ and we use $\mD = (R^\mD_J)_{J \in \mE}$ to denote an input database instance. The acyclicity notion for $\faq$s is identical to $\cq$s, through its associated hypergraph. Similar to $Q(\mD)$, we use $\varphi(\mD)$ to denote the query result of $\varphi(\bx_\mF)$ evaluated with input $\mD$, i.e. the resulting $\bS$-relation of schema $\bx_\mF$. A modification of Yannakakis algorithm (Algorithm~\ref{alg:yann}) can handle acyclic $\faq$s ~\cite{yannakakis1981algorithms}, by lifting the natural joins at line~\ref{line:j} to multiplications over the semiring domains, i.e.,
\begin{align} \label{eq:semiring-join}
    T^\mD_{\chi(s) \cup F_s} = R^\mD_{\chi(s)} (\bx_{\chi(s)})\otimes \bigoplus_{\chi(t) \setminus (\mF_t \cup (\chi(s) \cap \chi(t)))}T^\mD_{\chi(t)}(\bx_{\chi(t)}).
\end{align}


In this following, we show that our algorithm (\autoref{alg:our}) can be similarly extended to evaluate acyclic $\faq$s. We obtain the following result.

 \begin{theorem} \label{thm:faq-main}
     Given an acyclic $\faq$ query $\varphi(\bx_\mF)$ over a semiring $\bS$, database $\mD$, we can compute the output $\varphi(\mD)$ in time $O(|\mD| + |\tOUT|+|\mD| \cdot |\tOUT|^{1 - 1/\wout(Q)})$.
 \end{theorem}

The algorithm needs three simple augmentations from \autoref{alg:our}. First, for processing a heavy leaf at line~\ref{line:leaf}, call the modified Yannakakis algorithm for the sub-$\faq$. Second, use~\eqref{eq:semiring-join} again for line~\ref{line:join} instead of the natural joins. Lastly, we replace the unions for line~\ref{union1} and~\ref{union2} by $\oplus$ to aggregate back the query results of $\varphi(\bx_\mF)$. The correctness of this algorithm stems from the disjoint partitions of relations corresponding to each leaf on line~\ref{disjoint} of \autoref{alg:our}. In other words, it holds that $T^\mD_{\chi(s)} = T_{\chi(s)}^{\mD,H} \oplus T_{\chi(s)}^{\mD,L}$ and by distributivity, we have
$$
\varphi(\mD) = \varphi \left((\mD \setminus T^\mD_{\chi(s)} ) \cup T_{\chi(s)}^{\mD,H} \right) \oplus \varphi \left((\mD \setminus T^\mD_{\chi(s)}) \cup T_{\chi(s)}^{\mD,L} \right),
$$
where the first sub-query is evaluated upfront at line~\ref{call:lemone}. For the latter sub-query, if the leaf $s$ is the last leaf of its parent being processed, the sub-query is directly evaluated in the next for-loop iteration at its parent level (i.e. the else branch at line~\ref{line:nonleaf}). Otherwise, the relations in the database $(\mD \setminus T^\mD_{\chi(s)}) \cup T_{\chi(s)}^{\mD, L}$ will be further partitioned by the next sibling of $s$ in the post-order traversal. The runtime argument follows exactly from the proof of Theorem~\ref{thm:main}.


\section{Applications}

In this section, we we apply our framework to recover state-of-the-art results, as well as obtain new results, for queries of practical interest.

\subsection{Path Queries}

We will first study the \emph{path queries}, a class of queries that has immense practical importance. The projection width of a path query $P_k(x_1, x_{k+1})$ is $k$. Therefore, applying our main result, we get:

\begin{theorem} \label{thm:path}
    Given a path query $P_k(x_1, x_{k+1})$ and a database $\mD$, there exists an algorithm that can evaluate the path query in time $O(\revtwo{|\mD|} + |\mD| \cdot |\tOUT|^{1-1/k})$.
\end{theorem}

For $k=2$, our result matches the bound shown in~\cite{amossen2009faster}. Comparing our result to Hu, the bound obtained in~\autoref{thm:path} strictly improves the results obtained by Hu for $3 \leq k \leq 5$ and matches for $k=6$ even if we assume $\omega=2$. This result suggests there is room for further improvement in the use of fast matrix multiplication to obtain tighter bounds. For $k=7$, our running time is better than the one obtained by Hu assuming the current best known value of $\omega = 2.371552$~\cite{williams2024new}. 

\subsection{Hierarchical Queries}

We show here the application of our result to \emph{hierarchical queries}.  A CQ is {\em hierarchical} if for any two of its variables, either their sets of atoms are disjoint or one is contained in the other. All hierarchical queries are acyclic, and all star queries $Q^\star_\ell$ (defined in \autoref{ex:star}) are hierarchical queries. They have $\wout(Q^\star_\ell) = \ell$, thus:

\begin{theorem} \label{thm:star}
    Given a star query $Q^\star_\ell(x_{\ell})$ and a database $\mD$, the star query can be evaluated in time $O(|\mD|+ |\mD| \cdot |\tOUT|^{1-1/\ell})$.
\end{theorem}

\autoref{thm:star} recovers the bound from~\cite{amossen2009faster} for star queries and thus, provides an alternate proof of the combinatorial result in~\cite{amossen2009faster} for star queries. We note that for star queries the results merely improve the \emph{analysis} of the Yannakakis algorithm. In other words, Yannakakis algorithm also achieves the time bound as specified by \autoref{thm:star} but the routinely used upper bound of $O(|\mD| \cdot |\tOUT|)$ does not reflect that.

\begin{figure}
    \centering
    \scalebox{0.9}{
    \begin{tikzpicture}[xscale=1.15, yscale=1]
      \node at (0.0, 0.0) (A) {{\small \color{black} $x_1$}};
      \node at (-1.8, -0.8) (B) {{\small\color{black} $x_2$}} edge[-] (A);
      \node at (1.8, -0.8) (C) {{\small\color{black} $x_3$}} edge[-] (A);
      \node at (-2.7, -1.6) (D) {{\small\color{black} $x_4$}} edge[-] (B);
      \node at (-1, -1.6) (E) {{\small\color{black} $x_5$}} edge[-] (B);
      \node at (0.9, -1.6) (F) {{\small\color{black} $x_6$}} edge[-] (C);
      \node at (2.7, -1.6) (G) {{\small\color{black} $x_7$}} edge[-] (C);
      \node at (-2.7, -2.4) (R) {{\small \color{black} $R_{124}(\bx_{124})$}} edge[-] (D);
      \node at (-1, -2.4) (S) {{\small \color{black} $R_{125}(\bx_{125})$}} edge[-] (E);
      \node at (0.9,  -2.4)(T) {{\small \color{black} $R_{136}(\bx_{136})$}} edge[-] (F);
      \node at (2.7,  -2.4)(U) {{\small \color{black} $R_{137}(\bx_{137})$}} edge[-] (G);
    \end{tikzpicture} }
    \caption{Relations formed by variables arranged as a complete binary tree. Every root-to-leaf path forms a relation (labeled).} 
    \label{fig:hierarchy}
\end{figure}

As another example, consider the query $\revthree{Q(\bx_{4567})} \leftarrow R_{124}(\bx_{124}) \wedge R_{125}(\bx_{125}) \wedge R_{136}(\bx_{136}) \wedge R_{137}(\bx_{137})$ formed by the relations shown in \autoref{fig:hierarchy}. For this query, the projection width is four and thus, we obtain an evaluation time of $O(|\mD|+|\mD| \cdot |\tOUT|^{3/4})$.~\cite{kara2020trade} proposed an algorithm for enumerating the results of any hierarchical query (not necessarily full) with delay\footnote{The delay of enumerating query results refers to the upper bound on the time between outputting any two consecutive output tuples (including from start of the algorithm to the first tuple, and the last tuple to the end of the algorithm).} guarantees after preprocessing the input. \revthree{In particular, they showed that after preprocessing time $T_P = O(|\mD|^{1 + (\mw - 1) \cdot \epsilon})$, it is possible to enumerate the query result with delay $\delta = O(|\mD|^{1 - \epsilon})$, where $\mw$ is the \emph{static width} (a width parameter defined by~\cite{kara2020trade}) of a hierarchical query, for any $0 \leq \epsilon \leq 1$. Note that an algorithm with preprocessing $T_p$ and delay guarantee $\delta$ directly leads to a join evaluation algorithm that takes time $O(T_p + \delta \cdot |\tOUT|)$. For the example query $Q(\bx_{4567})$, it turns out that $\mw = 4$, and thus, the running time can be minimized for a suitable choice of threshold $\epsilon$ to also obtain the same running time that our algorithm achieves. A deeper exploration of this intriguing connection is a topic left for future research.}

\subsection{General Queries}

\introparagraph{Submodular width} A function $f : 2^{\nodes} \mapsto \mathbb{R}_+$ is a non-negative {\em set function} on $\nodes$ ($\ell \geq 1$).
The set function is {\em monotone} if $f(X) \leq f(Y)$ whenever $X \subseteq Y$, and is {\em submodular} if $f(X\cup Y)+f(X\cap Y)\leq f(X)+f(Y)$
for all $X,Y\subseteq \nodes$. A non-negative, monotone, submodular set function $h$ such that $h(\emptyset)=0$ is a \textit{polymatroid}. 
Let $Q$ be a CQ and let $\Gamma_{\ell}$ be the set of all polymatroids $h$ on $\nodes$ such that $h(J) \leq 1$ for all {$J \in \mE$}. The {\em submodular width} of $Q$ is
\begin{equation} \label{eq:subw}
    \subw(Q)  \defeq \; \max_{h \in \Gamma_{\ell}} \min_{(\mathcal{T}, \chi) \in \mathfrak{F}} \max_{t \in V(\mathcal{T})} h(\chi(t)),
\end{equation}
where $\mathfrak{F}$ is the set of all {\em non-redundant} tree decompositions of $Q$. A tree decomposition is {\em non-redundant} if no bag is a subset of another. {Abo Khamis et al.~\cite{PANDA} proved that non-redundancy ensures that $\mathfrak{F}$ is finite, hence the inner minimum is well-defined.}  Prior work~\cite{PANDA} showed that given any CQ $Q$ and database $\mD$, the \textsf{PANDA} algorithm can decompose the query and database instance into a constant number of pairs $(Q_i, \mD_i)$ such that $Q(D) \leftarrow \bigcup_i Q_i(\mD_i)$. Further, it is also guaranteed that each $Q_i$ is acyclic, $|\mD_i| = |\mD|^{\subw(Q)}$ (each $\mD_i$ can be computed in $\tilde{O}(|\mD|^{\subw(Q)})$ time), and $|Q_i(\mD_i)| \leq |Q(\mD)|$. Since each $Q_i$ is acyclic, we can apply our main result and obtain the following theorem. 

\begin{theorem} \label{thm:cyclic}
    Given a CQ $Q$ and database $\mD$, there exists an algorithm to evaluate $Q(\mD)$ in time $\tilde{O}(|\mD|^{\subw(Q)}+|\mD|^{\subw(Q)} \cdot |\tOUT|^{1-1/\max_i \wout(Q_i)})$, where $(Q_i, \mD_i)$ is the set of decomposed queries generated by $\mathsf{PANDA}$.
\end{theorem}

\begin{example} \label{ex:cyclic}
    Consider the 4-cycle query $Q^\diamond(\bx_{123}) \leftarrow R_{12}(\bx_{12}) \wedge R_{23}(\bx_{23}) \wedge R_{34}(\bx_{34}) \wedge R_{14}(\bx_{14})$. For this query , $\subw(Q^\diamond(\bx_{123})) = 3/2$ and $\mathsf{PANDA}$ partitions $Q^\diamond(\bx_{123})$ into two queries, $Q^\diamond_1(\bx_{123}) \leftarrow S_{123}(\bx_{123}) \wedge S_{134}(\bx_{134})$ and $Q^\diamond_2(\bx_{123}) \leftarrow S_{124}(\bx_{124}) \wedge S_{234}(\bx_{234})$. It is easy to see that $\wout(Q^\diamond_1) = 1$ and $\wout(Q^\diamond_2) = 2$. Thus, the query can be evaluated in time $\tilde{O}(|\mD|^{3/2} \cdot |\tOUT|^{1/2})$. We note that~\cite{hu2024fast} requires $\tilde{O}(|\mD|^{3/2} \cdot |\tOUT|^{5/6})$ time for the 4-cycle query $Q^\diamond(\bx_{123})$.
\end{example}

\section{A Faster Algorithm for Path Queries} \label{sec:newpath}

\begin{algorithm}[!ht]
    \DontPrintSemicolon 
    \SetKwInOut{Input}{Input}
    \SetKwInOut{Output}{Output}
    \Input{Path query $P_k$, database instance $\mD$}
    \Output{$P_k(\mD)$}
    \SetKwFunction{subquery}{\textsc{processSubquery}}
    \SetKwFunction{push}{\textsc{push}}
    \SetKwProg{myproc}{\textsc{procedure}}{}{}
    \SetKwData{ret}{\textbf{return}}
    $(\tree, \chi) \leftarrow $ join tree for $P_k$; $\mD := (R^\mD_{\chi(s)})_{s \in V(\tree)} \leftarrow $ \textbf{apply} a full reducer for $\mD$ \\
    $K, K' \leftarrow $ empty stack;  $N \leftarrow |\mD|; \mJ_1, \mJ_2, \mJ_3, \mJ_4 \leftarrow \emptyset$ \;
    \ForEach{$i \in \{ \lfloor   {k}/{2} \rfloor, \dots, 1\}$}{
        $K.\push(\chi^{-1}(\bx_{i,i+1}))$ \;
    }
    \ForEach{$i \in \{ \lfloor \frac{k}{2} \rfloor + 1, \dots, k \}$}{
        $K'.\push(\chi^{-1}(\bx_{i,i+1}))$
    }
    $\mJ_1, (\tree',\chi'), \mD' \leftarrow$ Call \autoref{alg:our} with $K$ as bag  order, $(\tree, \chi), \mD$, root node as bag $\chi^{-1}(\bx_{k,k+1})$, and threshold $\Delta$ \label{line:jone}\;
    \tcc{$R^\mD_{1,\lfloor   {k}/{2} \rfloor + 1}$ is the relation assigned to the leaf node in $\tree'$}
    $\Delta' \leftarrow |R^\mD_{1,\lfloor   {k}/{2} \rfloor + 1}(\bx_{1, \lfloor   {k}/{2} \rfloor + 1})| \cdot \Delta / |\tOUT|$ \;
    \tcc{Partition $R^\mD_{1,\lfloor   {k}/{2} \rfloor + 1}$ into heavy and light subrelations}
    $R^{\mD,H}_{1, \lfloor   {k}/{2} \rfloor + 1} = \{\bv \in R^\mD_{1, \lfloor   {k}/{2} \rfloor + 1} \mid |\sigma_{ \bv[\bx_1]}(R^\mD_{1, \lfloor   {k}/{2} \rfloor + 1})| > \Delta' \}$ \;
    $\mD^H \leftarrow (\mD' \setminus R^\mD_{1, \lfloor   {k}/{2} \rfloor + 1}) \cup R^{\mD,H}_{1, \lfloor   {k}/{2} \rfloor + 1}$ \;
    $\mJ_2 \leftarrow $ Call \autoref{alg:yann} on $(\tree',\chi')$ with root node as $\chi'^{-1}(\bx_{k, k+1})$ and $\mD^H$ \label{line:jtwo} \;
    $R^{\mD,L}_{1, \lfloor   {k}/{2} \rfloor + 1} = R^\mD_{1, \lfloor   {k}/{2} \rfloor + 1} \setminus R^{\mD,H}_{1, \lfloor   {k}/{2} \rfloor + 1}$ \;
    $\mD^L \leftarrow (\mD' \setminus R_{1, \lfloor   {k}/{2} \rfloor + 1}) \cup R^{\mD,L}_{1, \lfloor   {k}/{2} \rfloor + 1}$ \;
    $\mJ_3, (\tree'',\chi''), \mD'' \leftarrow$ Call \autoref{alg:our} with $K'$ as bag order, $(\tree', \chi'), \mD^L$, root node as bag $\chi'^{-1}(\bx_{1,\lfloor   {k}/{2} \rfloor + 1})$, and threshold $\Delta$ \label{line:jthree}\;
    $\mJ_4 \leftarrow \pi_{\bx_{1, k+1}} (R^{\mD,L}_{1, \lfloor   {k}/{2} \rfloor + 1}(\bx_{1, \lfloor   {k}/{2} \rfloor + 1}) \wedge R^\mD_{\lfloor   {k}/{2} \rfloor + 1, k+1}(\bx_{\lfloor   {k}/{2} \rfloor + 1, k+1}))$ \label{line:jfour}\;
    \ret $\bigcup_{i \in [4]} \mJ_i$
    \caption{Improved Path Query Evaluation} \label{alg:pathnew}
 \end{algorithm}

In this section, we show a better algorithm that improves upon~\autoref{thm:path} \revthree{by invoking \autoref{alg:our} in a novel way}.
\autoref{alg:pathnew} shows the steps for evaluating a path query $P_k$. The main idea of the algorithm is to carefully choose the ordering of how the nodes in a join tree are processed. This is in contrast with \autoref{alg:our} and the Yannakakis algorithm where any leaf-to-root order is sufficient to get join time guarantees. \revthree{The key insight of the algorithm is the following: the reader may observe that \autoref{alg:our} only partitions a leaf node relation $T_{\chi(s)}$ over $\chi^\Join(s)$. However, one could also partition a relation over the non-join variables (i.e. the isolated free variables) to further speed up query evaluation. We demonstrate that such a strategy can indeed be faster.}

\eat{Before we present the result, we note that the improved result is only applicable to path queries and its generalization to arbitrary acyclic CQs is non-trivial due to the branching structure in an acyclic CQ, which creates technical challenges for bounding intermediate result sizes. Nevertheless, the result is an exciting development for reasons outlined at the end of the section.}

\begin{theoremrep}\label{thm:newpath}
    Given query $P_k$ and a database $\mD$, there exists an algorithm to evaluate $P_k(\mD)$ in time $O(|\mD|+|\mD| \cdot |\tOUT|^{1 - 1/ \lceil (k+1)/2 \rceil})$ for any $k \geq 1$.
\end{theoremrep}
\begin{proof}
    We analyze the running time of \autoref{alg:pathnew}. It is easy to see that other than \autoref{line:jone}, \autoref{line:jtwo}, \autoref{line:jthree}, and \autoref{line:jfour}, all other steps take at most $O(|\mD|)$ time.

    To compute $\mJ_1$, it takes $O(|\mD| \cdot |\tOUT|/\Delta + |\mD| \cdot \Delta^{\lfloor   {k}/{2} \rfloor - 1})$ since $K$ restricts the processing to only $\lfloor   {k}/{2} \rfloor$ relations. After $\mJ_1$ has been computed, the leaf bag of the join tree $\tree'$ will contain variables $\{1, \lfloor   {k}/{2} \rfloor + 1\}$ and the assigned relation will be $R_{1, \lfloor   {k}/{2} \rfloor + 1}$. The size of the relation is $O(|\mD| \cdot \Delta^{\lfloor   {k}/{2} \rfloor - 1})$.
    
    Next, we partition relation $R^\mD_{1, \lfloor   {k}/{2} \rfloor + 1}$ on variable $x_1$ with degree threshold $\Delta ' = 
    |R^\mD_{1,\lfloor   {k}/{2} \rfloor + 1}| \cdot \Delta / |\tOUT|
    =O(|\mD| \cdot \Delta^{\lfloor   {k}/{2} \rfloor} / |\tOUT|)$. After partitioning, the active domain of $x_1$ in heavy-partition $R^{\mD,H}_{1, \lfloor   {k}/{2} \rfloor + 1}$ is \\ $O(|R^{\mD,H}_{1, \lfloor   {k}/{2} \rfloor + 1}| / \Delta') = O(|\tOUT| / \Delta)$ and thus, calling Yannakakis algorithm on~\autoref{line:jtwo} takes time $O(|\mD| \cdot |R^{\mD,H}_{1, \lfloor   {k}/{2} \rfloor + 1}| / \Delta') = O(|\mD| \cdot |\tOUT| / \Delta)$.

    At this point, all values for variable $x_1$ in $R^{\mD,L}_{1, \lfloor   {k}/{2} \rfloor + 1}$ have degree at most $\Delta'$. The next call on~\autoref{line:jthree} requires time $O(|\mD| \cdot |\tOUT|/\Delta + |\mD| \cdot \Delta^{\lceil   {k}/{2} \rceil - 1})$. This is because except the root node, all other relations are of size $O(|\mD|)$ and thus, when \autoref{alg:our} invokes \autoref{lem:basic}, $R_{1, \lfloor   {k}/{2} \rfloor + 1}$ becomes a leaf node and  its size is not used in the running time expression (recall that \autoref{lem:basic} only uses the sum of sizes of internal nodes of the join tree in its time complexity analysis). Further, observe that by the choice of $K'$, when the algorithm terminates, the join tree $\tree''$ returned will contain $R_{\lfloor   {k}/{2} \rfloor + 1, k}$ as the leaf relation and $R_{1, \lfloor k /2 \rfloor + 1}$ as the root relation.

    Finally, \autoref{line:jfour} joins the two remaining relations in $\tree''$ and returns the output to the user. The key insight here is that since $x_1$ is light, the join time can be bound as $O(|\tOUT| \cdot \Delta') = O(|\mD| \cdot \Delta^{\lfloor   {k}/{2} \rfloor})$. Putting all things together, the total running time is big-O of
    $$ |\mD| \cdot |\tOUT|/\Delta + |\mD| \cdot \Delta^{\lceil   {k}/{2} \rceil - 1} + |\mD| \cdot \Delta^{\lfloor   {k}/{2} \rfloor}$$

    By using the fact that $\lceil   {k}/{2} \rceil - 1 \leq \lfloor   {k}/{2} \rfloor$ and minimizing the expression, we get the optimal value of $\Delta = |\tOUT|^{1 / (\lfloor   {k}/{2} \rfloor + 1)} = |\tOUT|^{1 / \lceil   (k+1)/{2} \rceil}$ since $\lceil   (k+1)/{2} \rceil = \lfloor   {k}/{2} \rfloor + 1$ for any positive integer $k$. Observe that for the optimal choice of $\Delta$, $1 \leq \Delta' = |\mD| / |\tOUT|^{1 / (1 + \lfloor k/2 \rfloor)} \leq |\mD|$ for all values of $k$ and $|\tOUT|$, and therefore is well-defined. Thus, the total running time of the algorithm is $O(|\mD| \cdot |\tOUT|^{1 - 1/ \lceil   (k+1)/{2} \rceil})$.
\end{proof}

\begin{example}
    Consider the $P_4(x_1, x_5) \leftarrow R_{12}(\bx_{12}) \wedge R_{23}(\bx_{23})\wedge R_{34}(\bx_{34}) \wedge R_{45}(\bx_{45})$. The algorithm sets $K = \{\chi^{-1}(\bx_{12}), \chi^{-1}(\bx_{23})\}$ and $K' = \{\chi^{-1}(\bx_{45}), \chi^{-1}(\bx_{34})\}$. First, \autoref{line:jone} computes the join output when variable $x_2$ is heavy and stores it in $\mJ_1$ and the returned join tree $\tree'$ contains $\bx_{13}$ as the leaf node bag with a materialized relation that corresponds to the join of $R^\mD_{12}(\bx_{12})$ and $R^\mD_{23}(\bx_{23})$ but when $x_2$ is light in $R_{12}$. In other words, we  get the relation $R^\mD_{13}$ of size $|\mD| \cdot \Delta$. Then, we partition $R^\mD_{13}$ into two sub-relations: $R^{\mD,H}_{13}$ and $R^{\mD,L}_{13}$ based on degree of $x_1$ with degree threshold $\Delta' = |\mD| \cdot \Delta^2 / |\tOUT|$. Once the partitioning is done, we compute $\mJ_2 = \pi_{x_1, x_5} (R^{\mD,H}_{13}(\bx_{13}) \wedge R^\mD_{34}(\bx_{34}) \wedge R^\mD_{45}(\bx_{45}))$ using Yannakakis algorithm. In the third step, we take the join tree of the subquery $\pi_{x_1, x_5} (R^{\mD,L}_{13}(\bx_{13}) \wedge R^\mD_{34}(\bx_{34}) \wedge R^\mD_{45}(\bx_{45}))$ rooted at bag $\bx_{13}$ and generate output $\mJ_3$ when $x_4$ is heavy in relation $R^\mD_{45}$. This step takes $O(|\mD| \cdot |\tOUT| / \Delta)$ time according to \autoref{lem:basic}. For all light $x_4$ values in $R^\mD_{45}$, \autoref{alg:our} will compute the join of $R^\mD_{34}(\bx_{34})$ and $R^\mD_{45}(\bx_{45})$, and store the materialized result in $R^\mD_{35}(\bx_{35})$ in time $O(|\mD| \cdot \Delta)$. The final step is to compute $\pi_{x_1, x_5}(R^{\mD,L}_{13}(\bx_{13}) \wedge R^\mD_{35}(\bx_{35}))$ which takes $O(|\tOUT| \cdot \Delta')$ time. Balancing all the costs, we obtain $\Delta = |\tOUT|^{1/3}$ and a total running time of $O(|\mD| \cdot |\tOUT|^{2/3})$.
\end{example}

\autoref{thm:newpath} has interesting implications in the evaluation of path queries. Observe that for $k=3$, we get the running time as $O(|\mD| + |\mD| \cdot \sqrt{|\tOUT|})$, matching the lower bound shown by Hu. Using the same ideas from~\autoref{sec:agg}, it is straightforward to adapt \autoref{alg:pathnew} to allow FAQs over path queries as well.
\section{Lower Bounds}

\looseness=-1
In this section, we demonstrate several lower bounds that show optimality for a subclass of CQs. \revthree{First, we define the $k$-clique problem that will be central to our lower bounds. Given an undirected graph $G = (V, E)$ and an integer $k \leq |V|$, the $k$-clique problem consists of deciding if the graph $G$ contains $k$ vertices such that each of the $k$ vertices are connected to each other via an edge in $E$. We also define the \emph{minimum-weight} $k$-clique problem: suppose the graph $G$ is equipped with an edge-weight function that maps edges to weights in $[0, M]$ for integer $M > 0$, find the $k$-clique where the edge sum is minimized.} We will use the following well-established conjectures from fine-grained complexity for the two $k$-clique problems.

\begin{definition}[Boolean $k$-Clique Conjecture]
There is no real $\epsilon >0$ such that computing the $k$-clique problem (with $k \geq 3$) over the Boolean semiring in an (undirected) $n$-node graph requires time $O(n^{k-\epsilon})$ using a combinatorial algorithm.
\end{definition}

\begin{definition}[Min-Weight $k$-Clique Conjecture]
 There is no real $\epsilon >0$ such that computing the $k$-clique problem (with $k \geq 3$) over the tropical semiring in an (undirected) $n$-node graph with integer edge weights can be done in time $O(n^{k-\epsilon})$.
\end{definition}

Our first lower bound tells us that the dependence in the bound of~\autoref{thm:cyclic} on both the submodular width and projection width is somewhat necessary.

\begin{theoremrep}\label{thm:sub:lowerbound}
Take any integer $\ell \geq 2$ and any rational $w \geq  1$ such that $\ell \cdot w$ is an integer. Then, there exists a query $Q$ with projection width \revthree{$\ell$ free variables} and submodular width $w$ such that no combinatorial algorithm can compute it over input $\mD$ in time $O(|\mD|^{w} \cdot |\tOUT|^{1-1/\ell - \epsilon})$ \revthree{ for any real $\epsilon > 0$}, assuming the Boolean $k$-Clique Conjecture.
\end{theoremrep}
\begin{proof}
Consider the graph $G = (V,E)$ for which we need to decide whether there exists a clique of size $k = (\ell-1)+\ell \cdot w$. Let $n = |V|$. It will be helpful to distinguish the $k$ variables of the clique into $\ell$ variables $x_1, \dots, x_\ell$ and the remaining $k-\ell$ variables $y_1, \dots, y_{k-\ell}$. Note that $k-\ell = \ell \cdot w -1$ is an integer $\geq 1$; hence, there exists at least one $y$-variable.

We will first compute all the cliques of size $\ell$ in $G$ and store them in a relation $C(v_1, \dots, v_\ell)$; this forms the input database $\mD$. Note that this step needs $O(n^\ell)$ time, and moreover the size of $C$ is $n^\ell$. Now, consider the following query:
$$ Q (x_1, \dots, x_\ell) \leftarrow U_1[x_1, y_1, \dots, y_{k-\ell}] \wedge U_2[x_2, y_1, \dots, y_{k-\ell}] \wedge \dots \wedge U_\ell[x_\ell, y_1, \dots, y_{k-\ell}]$$
where $U_i[x_i, y_1, \dots, y_{k-\ell}]$ stands for the join of following atoms:
$$ \{C(v_1, \dots, v_\ell) \mid v_1, \dots, v_\ell \text{ are all possible sets of size $\ell$ from } \{x_i, y_1, \dots, y_{k-\ell}\} \}$$
This is a well-defined expression, since $k-\ell = \ell \cdot w -1$ and $w \geq 1$. The submodular width of this query is $w$, as proved in~\cite{zhao2024evaluating}. It is also easy to see that \revthree{ the query has $\ell$ free variables}.

Suppose now we can compute $Q$ in time $O(|\mD|^{\subw} \cdot |\tOUT|^{1-1/\ell-\epsilon} )$ for some $\epsilon > 0$. Note that $|\mD| \leq n^\ell$ and $|\tOUT| \leq n^\ell$. Hence, this implies that we can compute $Q$ on the above input in time $O(n^{w \cdot \ell + \ell - 1 - \ell \epsilon})= O(n^{k-\epsilon'})$ \revthree{where $\epsilon' = \ell \epsilon$}. 

Once the output has been evaluated, for each tuple $t \in Q (x_1, \dots, x_\ell)$, we can verify in constant time that any two vertices $t(i)$ and $t(j)$ ($1 \leq i, j \leq \ell$) are connected via an edge. Thus, in $O(n^\ell)$ time, we can verify whether the vertices form an $\ell$-clique. Further, by our join query construction, each of the $\ell$ vertices is guaranteed to connect with a common vertex $v'$ (which corresponds to a valuation of the remaining variables $y_1, \dots, y_{k-\ell}$). 

All together, this obtains a (combinatorial) algorithm that computes whether a $k$-clique exists in time $O(n^{k-\epsilon'} + n^\ell) = \revthree{O(n^{k-\epsilon'})}$, contradicting the lower bound conjecture.
\end{proof}

We next prove a general result that gives output-sensitive lower bounds for arbitrary CQs. To this end, we utilize the notion of \emph{clique embedding}~\cite{FanKZ23}. We say that two sets of vertices $X,Y \subseteq V(H)$ {\em touch} in $\mH$ if either $X \cap Y \neq \emptyset$ or there is a hyperedge $e \in E(\mH)$ that intersects both $X$ and $Y$.

\begin{definition}[Clique Embedding]
Let $k \geq 3$ and $\mH$ be a hypergraph. A {\em $k$-clique embedding}, denoted as $C_k \mapsto \mH$, is a mapping $\psi$ that maps every $v \in \{1, \dots, k\}$ to a non-empty subset $\psi(v) \subseteq V(\mH)$ such that the following hold:
\begin{enumerate}
\item $\psi(v)$ induces a connected subgraph; 
\item for any two $u \neq v \in \{1, \dots, k\}$ then $\psi(u), \psi(v)$ {\em touch} in $\mH$.
\end{enumerate}
\end{definition}

It is often convenient to describe a clique embedding $\psi$ by the reverse mapping $\psi^{-1}(x) = \{i \mid x \in \psi(i) \}$, for $x \in V(\mH)$. 
For a hyperedge $J \in E(\mH)$, its {\em weak edge depth} is $d_{\psi}(J) := |\{v \mid  \psi(v) \cap J \neq \emptyset \}|$, i.e., the number of vertices from $V(\mH)$ that map to some variable in $J$. We also define the {\em edge depth} of $J\in E(\mH)$ as $d^+_{\psi}(J) := \sum_{v \in J} d_\psi(v)$, i.e. weak edge depth but counting multiplicity. 
For an embedding $\psi$ and a hyperedge $\mF \in E(\mH)$, the {\em $\mF$-weak edge depth of $\psi$} is defined as $\mathsf{wed}^\mF(\psi):= \max\limits_{J \in E(\mH) \setminus \mF} d_{\psi}(J)$, i.e. the maximum of weak edge depths excluding $\mF$.


\begin{theoremrep}\label{thm:boolean_lower_bound}
    Given any CQ $Q(\bx_\mF) \leftarrow \bigwedge_{J \in \mE} R_J(\mathbf{x}_J)$ and database $\mD$, 
    let $\psi$ be a $k$-clique embedding of the hypergraph $\mH' = ([n],E(\mH) \cup \{\mF\})$ such that $x\cdot \mathsf{wed}^\mF(\psi) + y \cdot d^+_{\psi}(\mF) \leq k$ for positive $x,y >0$.
    Then, there is no combinatorial algorithm with running time $O(|\mD|^{x-\epsilon} \cdot |\tOUT|^{y-\epsilon'})$ for evaluating $Q(\mD)$ assuming the Boolean $k$-Clique Conjecture, where $\epsilon, \epsilon'$ are non-negative with $\epsilon + \epsilon' > 0$.
\end{theoremrep}
\begin{proof}
    Construct the instance as Theorem 11 in~\cite{FanKZ23} and observe that there exists a $k$-clique if and only if there exists a tuple in the output that is consistent with the domains in variables $\mF$. By a consistent tuple, we mean a tuple whose values associated to the same partitions of the input graph for finding $k$-clique are the same. The time and size to construct such instance $\mD$ is bounded by $O(n^{\mathsf{wed}^\mF(\psi)})$, and observe the size of $|\tOUT|$ is bounded by  $O(n^{d^+_{\psi}(\mF)})$. Thus, such output-sensitive combinatorial algorithm will yield a combinatorial algorithm for $k$-clique that runs in time
    $O(n^{x\cdot (\mathsf{wed}^\mF(\psi)-\epsilon)} \cdot n^{y\cdot (d^+_{\psi}(\mF)-\epsilon')}) = O(n^{k-\epsilon''})$ for some $\epsilon'' >0$, contradicting the Boolean $k$-Clique Conjecture.
\end{proof}
We can show an analogous result for $\faq$ queries over the tropical semiring.

\begin{theoremrep}\label{thm:trop_lower_bound}
    Given any $\faq$ $\varphi(\bx_\mF)$ as~\eqref{eq:sumprod-rule} and database $\mD$, let $\psi$ be a $k$-clique embedding of the hypergraph $\mH' = ([n],E(\mH) \cup \{\mF\})$ such that $x\cdot \mathsf{wed}^\mF(\psi) + y \cdot d^+_{\psi}(\mF) \leq k$ for positive $x,y >0$. Then, there exists no algorithm that evaluates $\varphi(\mD)$ over the tropical semiring in time $O(|\mD|^{x-\epsilon} \cdot |\tOUT|^{y-\epsilon'})$ assuming the Min-Weight $k$-Clique Conjecture, where $\epsilon, \epsilon'$ are non-negative with $\epsilon + \epsilon' > 0$.
\end{theoremrep}
\begin{proof}
    Construct the instance as Theorem 11 in~\cite{FanKZ23}. Given the output tuples, we simply check for the consistency of domains in variables $\bx_\mF$ and return the smallest annotation among consistent tuples after sorting. Observe that this procedure correctly solves the Min-Weight $k$-clique problem. The time analysis is then identical as the proof of Theorem~\ref{thm:boolean_lower_bound}.
\end{proof}




\looseness=-1 Hu defined the notion of $\sfreew$ of any acyclic query and proved that any "semiring algorithm" requires $\Omega(|\mD|\cdot|\tOUT|^{1-\frac{1}{\sfreew(Q)}} + |\tOUT|)$ ~\cite{hu2024fast}. Our lower bound is complementary to Hu's lower bound. Specifically, Hu showed a stronger result that the lower bound applies for every value of $|\mD|$ and $|\tOUT|$, whereas our result shows the existence of a hard database instance. However, our lower bound also applies to cyclic CQs whereas Hu's $\sfreew(Q)$ is not defined for cyclic queries. Therefore, the two results are incomparable. For the special case of $x=1$, we match Hu's lower bound (see Proposition~\ref{prop:Hu}) by providing an alternate proof that there exists a clique embedding for $Q'$ such that $y = 1-\frac{1}{\sfreew(Q)}$. 
We now apply Theorem~\ref{thm:boolean_lower_bound} and Theorem~\ref{thm:trop_lower_bound} to get tight lower bounds for star queries and can be extended to \revthree{hierarchical queries (note that star queries are also hierarchical queries).}

\begin{toappendix}
\begin{proposition}\label{prop:Hu}
    Given any CQ $Q(\bx_\mF) \leftarrow \bigwedge_{J \in \mE} R_J(\mathbf{x}_J)$, define $Q'(\bx_\mF) \leftarrow \bigwedge_{J \in \mE} R_J(\bx_J) \wedge R(\bx_\mF)$. Then, there exists a $k$-clique embedding $\psi$ such that $\frac{k-\mathsf{wed}^\mF(\psi)}{d^+_{\psi}(\mF)} = 1 - \frac{1}{\sfreew(Q)}$.
\end{proposition}
\begin{proof}
    It suffices to construct a clique embedding on the connected components that achieves $\sfreew(Q)$ after Hu's cleanse and decomposition steps~\cite{hu2024fast}. It is straightforward to see that one can assign each element in $S$ an isolated output variable~\cite{hu2024fast}. Consider the $(\sfreew(Q)+1)$-clique embedding $\psi$ where $\psi$ maps $\sfreew(Q)$ distinct colors to each one of those distinguished isolated output variables, and map an additional fresh color to all other variables (including all join variables and possibly some isolated output variables that are not assigned to $S$). Clearly, $\psi$ is a $(\sfreew(Q)+1)$-clique embedding for $Q'$. Moreover, $\mathsf{wed}^\mF(\psi) = 2$ and $d^+_{\psi}(\mF) = \sfreew(Q)$, and thus  $\frac{k-\mathsf{wed}^\mF(\psi)}{d^+_{\psi}(\mF)} = \frac{\sfreew(Q)+1 - 2}{\sfreew(Q)} = 1-\frac{1}{\sfreew(Q)}$, completing the proof. 
\end{proof}
\end{toappendix}

\begin{theoremrep}\label{thm:boolean}
For the star query $Q^\star_\ell$, there exists a database $\mD$ such that no algorithm can have runtime of $O(|\mD| \cdot |Q^\star_\ell(\mD)|^{1-1/\ell-\epsilon})$ \revthree{ for any real $\epsilon > 0$} subject to the Boolean $k$-clique conjecture.
\end{theoremrep}
\begin{proof}
The star query $Q^\star_\ell$ is defined as $Q^\star_\ell (x_1, \dots, x_\ell) \leftarrow R_1(x_1,y) \wedge \dots \wedge R_\ell(x_\ell, y)$. Consider the query $(Q^\star_\ell)' (x_1, \dots, x_\ell) \leftarrow R_1(x_1,y) \wedge \dots \wedge R_\ell(x_\ell, y) \wedge R(x_1, x_2, \dots, x_\ell)$. Observe that $\psi(i) = x_i$ for $1\leq i \leq \ell$ and $\psi(\ell+1) = y$ is a $(\ell+1)$-clique embedding of $(Q^\star_\ell)'$. This has $\mathsf{wed}^\mF(\psi) = 2$ and $d^+_{\psi}(\mF) = \ell $. For $x=1$ and $y = 1-1/\ell$, we have $x\cdot \mathsf{wed}^\mF(\psi) + y \cdot d^+_{\psi}(\mF) \leq \ell+1$.
 Now apply Theorem~\ref{thm:boolean_lower_bound}.
\end{proof}

\begin{restatable}{theorem}{startropicallb}\label{thm:tropical}
    For the star query $Q^\star_\ell$, there exists a database $\mD$ such that no algorithm can have runtime of $O(|\mD|\cdot |Q^\star_\ell(\mD)|^{1-1/\ell-\epsilon}) $ \revthree{for any real $\epsilon > 0$} subject to the Min-Weight $k$-Clique Conjecture.
\end{restatable}

\eat{While star query evaluation is optimal as shown by our lower bounds, results in Section~\ref{sec:newpath} show that it is possible to further improve the query evaluation runtime for other subclasses of queries. We leave the problem of improving both upper and lower bounds for arbitrary CQ as an open problem.}




\eat{The following theorem gives conditional lower bounds for $\ell$-paths for any integer $\ell\geq 2$ assuming the Boolean $k$-clique conjecture. This can be seen as an output-sensitive version of the lower bound results in~\cite{FanKZ23}. At a high level, we ``close'' the $\ell$-path into a $\ell$-cycle by examining the join output. For the ease of readers, we don't introduce the whole machinery of clique embedding power in~\cite{FanKZ23} but provide a self-contained proof specialized to the case of $\ell$-path.

\austen{When plugged into $\ell = 4$, this gives $\Omega(|\mD| \cdot |\tOUT|^{1/2})$ lower bound for 3-path, which tightly matched Shaleen's new algorithm. When $\ell = 5$, this gives $\Omega(|\mD| \cdot |\tOUT|^{2/3})$ lower bound for 4-path, which is also tight due to the new discovery. However, when $\ell = 6$, i.e. 5-path, we don't have a matching upper bound yet.} 

\begin{theorem}
    For any $(\ell-1)$-path query $Q(x_1,x_\ell):- U_1(x_1,x_2), U_2(x_2,x_3), \dots, U_{\ell-1}(x_{\ell-1},x_\ell)$, there exists no combinatorial algorithm in time $O(|\mD|\cdot |\tOUT|^{1 - 1 / \lceil \ell / 2 \rceil - \epsilon})$ for any $\epsilon > 0$ assuming the Boolean $k$-clique conjecture.
\end{theorem}

\begin{proof}
We start with the case when $\ell$ is odd. Let $\lambda = (\ell+1)/2$. Consider the following mapping $\psi$ from $[\ell]$ to variables $x_i$'s, $1 \leq i \leq \ell$, uniquely defined by the following:
    \begin{equation}
    \begin{aligned}
     \psi^{-1}(x_1) & = \{ 1,2, \dots, \lambda-1\} \\ 
      \psi^{-1}(x_2) & = \{ 2, 3, \dots, \lambda\} \\ 
      & \dots \\
    \psi^{-1}(x_\ell) & = \{ 2\lambda-1, 1, \dots, \lambda-2\}
    \end{aligned}
\end{equation}

In other words, $\psi$ maps each $i \in [\ell]$ into a consecutive segment consisting of $\lambda-1$ vertices in the $\ell$-cycle, if we were to add the atom $U_{\ell} (x_\ell, x_1)$. Observe that for any $i,j \in [\ell]$, there exists $p \in [\ell]$ such that $\{i,j\} \in \psi^{-1}(x_p) \cap \psi^{-1}(x_{p+1})$. Also, observe that for any $p \in [\ell]$, we have $|\psi^{-1}(x_p) \cup \psi^{-1}(x_{p+1})| = \lambda = (\ell+1) /2$.

Without loss of generality, we can assume that the input graph for deciding $\ell$-clique is $\ell$-partite. Indeed, given any graph $G = (V,E)$ where $V = \{v_1, v_2, \dots, v_n\}$, consider the $\ell$-partite graph $G^\ell = (V^\ell,E^\ell)$ where $V^\ell = \{v_i^j \mid 1 \leq i \leq n, 1 \leq j \leq k\}$ and for any two vertices $v_i^j, v_p^q \in V^k$, $\{v_i^j, v_p^q\} \in E^k $ iff $\{v_i, v_p\} \in E$ and $j \neq q$. Then there is a one-to-one mapping from a $\ell$-clique in $G$ to a $\ell$-clique in $G^\ell$.

We now construct an instance for $Q(x_1, x_\ell)$. The domain of each variable $x_i$ is the set of vectors over $[n]^{|\psi^{-1}(x_i)|}$, where  $n$ is the number of input graph $G$. Let $S_e = \{ i\in [\ell] \mid \psi(i) \cap e \neq \emptyset \}$ where $e = (x_p, x_{p+1})$ for some $p\in [\ell-1]$. By our construction, we have $|S_e| \leq \lambda$. We now compute all cliques in graph $G$ between the partitions $V_i$ where $i \in S_e$; these cliques will be of size $|S_e|$ and can be computed in time $O(n^\lambda)$ by brute force. Finally, for every clique $\{a_i \in V_i \mid i \in S_e\}$, add the tuple $t:= \prod_{i \in S_e} \mathsf{Dom}(x_i)$ where the value for variable $x$ is $\langle a_i \mid i \in \psi^{-1}(x) \rangle$ into the atom corresponding to $e$. Observe that the size of the instance $|\mD|$ and output $|\tOUT|$ are both bounded by $O(n^\lambda)$.

When $\ell$ is odd, we have $\lceil \ell / 2 \rceil = (\ell + 1)/2 = \lambda$. Assume that we have a combinatorial algorithm running in time $O(|\mD|\cdot |\tOUT|^{1 - 1 / \lambda - \epsilon})$ for some $\epsilon >0$. Then, we can solve this instance in time $O(n^\lambda \cdot n^{\lambda\cdot(1-1/\lambda - \epsilon )}) = O(n^{\lambda\cdot (2-1/\lambda) - \epsilon'}) = O(n^{2\lambda-1-\epsilon'}) = O(n^{\ell-\epsilon'}).$ Finally, we  check if there exists a tuple in the output that has identical entries on $\psi^{-1}(x_1)$ and $\psi^{-1}(x_\ell)$ and this can be done in $|\tOUT| = O(n^\lambda)$ time. Observe that such a tuple exists if and only if there exists an $\ell$-clique in $G$. This contradicts the Boolean $k$-clique conjecture. 

The case when $\ell$ is even is done by changing the mapping $\psi$ to be from $[\ell-1]$ to variables $x_i$'s, $1 \leq i \leq \ell$, uniquely defined by the following:
\begin{align*}
 \psi^{-1}(x_1) & = \{ 1,2, \dots, \lambda-1\} \\ 
  \psi^{-1}(x_2) & = \{ 2, 3, \dots, \lambda\} \\ 
  & \dots \\
   \psi^{-1}(x_{\ell - 1}) & = \{ 2\lambda-2, 2\lambda-1, 1, \dots, \lambda-3\}  \\
\psi^{-1}(x_\ell) & = \psi^{-1}(x_{\ell - 1})
\end{align*}
The proof is now complete by following the same construction when $\ell$ is odd.
\end{proof}

\austen{This can generalize to any semiring. Also, this can generalize to $|\mD| \cdot |\tOUT|^{\textsf{clemb}(Q')-1}$ where $Q'$ is $Q$ union  an atom with $Q$'s free variables. See sketch in the appendix.}}
\section{Related Work}

Several prior works have investigated problems related to output sensitive evaluation. Deng et al.~\cite{deng2023join} presented a dynamic index structure for output sensitive join sampling. Riko and St{\"o}ckel~\cite{jacob2015fast} studied output-sensitive matrix multiplication (which is the two path query), a result that built upon the join-project query evaluation results from~\cite{amossen2009faster}. Deng et al.~\cite{deng2018overlap} presented output-sensitive evaluation algorithms for set similarity, an important practical class of queries used routinely in recommender systems, graph analytics, etc. An alternate way to evaluate join queries is to use {delay} based algorithms. \eat{In particular, given an algorithm with preprocessing time $T_p$ and allows for $\delta$ delay enumeration  of the query result, one can obtain an evaluation algorithm that takes time $O(T_p + \delta \cdot |\tOUT|)$ time.} As we saw in Section~\ref{sec:agg}, our result is able to match the join running time obtained via~\cite{kara2020trade} for star queries. Recent work~\cite{agarwal2024reporting} has also studied the problem of reporting $t$ patterns in graphs when the input is temporal, i.e., the input is changing over time. Output-sensitive algorithms has also been developed for the problem of maximal clique enumeration~\cite{chang2013fast, makino2004new,conte2016sublinear}. Several interesting results have been known for listing $k$-cliques and $k$-cycles in output-sensitive way.~\cite{bjorklund2014listing} designed output-sensitive algorithms
for listing $t$ triangles, the simplest $k$-clique, using fast matrix multiplication.The work of~\cite{jin2023removing} and~\cite{abboud2022listing} show algorithms for list $t$ 4-cycles, and~\cite{jin2024listing} shows how to list $t$ 6-cycles. The upper bounds were shown to be tight under the 3SUM hypothesis.
\section{Conclusion}
In this paper, we presented a novel generalization of Yannakakis algorithm that is provably faster for a large class of CQs. We show several applications of the generalized algorithm (which is combinatorial in nature) to recover state-of-the-art results known in existing literature, as well as new results for popular queries such as star queries and path queries. \eat{At the heart of our technical development is a novel lemma that allows for a tighter bound on the intermediate result size as we process the join tree in a bottom-up fashion. }Surprisingly, our results show that for several queries, our combinatorial algorithm is better than known results that use fast matrix multiplication. We complement our upper bounds with a matching lower bound for a \revthree{subclass of cyclic and acyclic CQs}, for both Boolean semirings (a.k.a joins) as well as aggregations.

\bibliographystyle{ACMReferenceFormat}
\bibliography{ref}

\end{document}